\documentclass[10pt,a4paper]{article}
\usepackage{amssymb}
\usepackage{amsmath}
\usepackage{amsthm}
\usepackage{latexsym}
\usepackage[dvips]{epsfig}
\usepackage{enumerate}
\usepackage{mathrsfs}
\usepackage{eufrak}
\usepackage{bm}
\usepackage{tikz}
\usepackage{authblk}

\usepackage{cleveref}

\theoremstyle{plain}
\newtheorem{proposition}{Proposition}
\newtheorem{lemma}{Lemma}
\newtheorem{theorem}{Theorem}
\newtheorem{assumption}{Assumption}

\newtheorem*{main}{Theorem}
\newtheorem{definition}{Definition}
\newtheorem{remark}{Remark}

\def\bmg{{\bm g}}

\def\bmF{{\bm F}}

\def\bmL{{\bm L}}
\def\bmM{{\bm M}}

\def\bmQ{{\bm Q}}

\def\bmT{{\bm T}}

\def\bmUpsilon{{\bm \Upsilon}}

\def\bmLambda{{\bm \Lambda}}

\def\fraka{\mathfrak{a}}
\def\frakb{\mathfrak{b}}
\def\frakc{\mathfrak{c}}
\def\frakd{\mathfrak{d}}
\def\frake{\mathfrak{e}}

\newcounter{mnotecount}

\newcommand{\mnotex}[1]
{\protect{\stepcounter{mnotecount}}$^{\mbox{\footnotesize $\bullet$\themnotecount}}$ 
\marginpar{
\raggedright\tiny\em
$\!\!\!\!\!\!\,\bullet$\themnotecount: #1} }

\begin{document}

\title{\textbf{Conformal wave equations for the Einstein-tracefree
    matter system}}

\author[1]{Diego A. Carranza \footnote{E-mail address:{\tt d.a.carranzaortiz@qmul.ac.uk}}}
\author[1]{Adem E. Hursit \footnote{E-mail address:{\tt a.e.hursit@qmul.ac.uk}}}
\author[1]{Juan A. Valiente Kroon \footnote{E-mail address:{\tt j.a.valiente-kroon@qmul.ac.uk}}}
\affil[1]{School of Mathematical Sciences, Queen Mary, University of London,
Mile End Road, London E1 4NS, United Kingdom.}

\medskip

\maketitle

\begin{abstract}
Inspired by a similar analysis for the vacuum conformal Einstein field
equations by Paetz [Ann. H. Poincar\'e \textbf{16}, 2059 (2015)], in this
article we show how to construct a system of quasilinear wave equations for the
geometric fields associated to the conformal Einstein field equations coupled
to matter models whose energy-momentum tensor has vanishing trace. In this
case, the equation of conservation for the energy-momentum tensor is
conformally invariant. Our analysis includes the construction of a subsidiary
evolution which allows to prove the propagation of the constraints. We discuss
how the underlying structure behind these systems of equations is the set of
integrability conditions satisfied by the conformal field equations. The main
result of our analysis is that both the evolution and subsidiary equations for
the geometric part of the conformal Einstein-tracefree matter field equations
close without the need of any further assumption on the matter models other
that the vanishing of the trace of the energy-momentum tensor. Our work is
supplemented by an analysis of the evolution and subsidiary equations
associated to three basic tracefree matter models: the conformally invariant
scalar field, the Maxwell field and the Yang-Mills field. As an application we
provide a global existence and stability result for de Sitter-like spacetimes.
In particular, the result for the conformally invariant scalar field is new in
the literature. 
\end{abstract}

\section{Introduction}
The \emph{conformal Einstein field equations} are a conformal representation of
the Einstein field equations which permit us to study the global properties of the
solutions to equations of General Relativity by means of Penrose's procedure of
conformal compactification --- see e.g. \cite{Fri15,CFEBook} for an entry point
to the literature on the subject. Crucially, a solution to the conformal
Einstein field equations implies a solution to the Einstein field equations away
from the conformal boundary. 

\medskip
A key step in the analysis involving the conformal Einstein field equations is
the so-called \emph{procedure of hyperbolic reduction}, in which a subset of
the field equations is cast in the form of a hyperbolic evolution system (the
\emph{evolution system}) for which known techniques of the theory of partial
differential equations allow us to establish well-posedness. An important
ingredient in the hyperbolic reduction is the choice of a gauge, which in the
case of the conformal Einstein field equations involves not only fixing
coordinates (the \emph{coordinate gauge}) but also the representative of the
conformal class of the spacetime metric (the so-called \emph{unphysical
metric}) to be considered (the \emph{conformal gauge}). Naturally, gauge
choices should bring to the fore the physical and geometric features of the
setting under consideration. In order to make contact with the Einstein field
equations, the procedure of hyperbolic reduction has to be supplemented by an
argument concerning the \emph{propagation of the constraints}, by means of
which one identifies the conditions under which one can guarantee that a
solution to the evolution system implies a solution to the full system of
conformal equations, independently of the gauge choice. The propagation of the
constraints involves the construction of a \emph{subsidiary evolution system}
describing the evolution of the conformal field equations and of the conditions
representing the gauge. The construction of the subsidiary system requires
lengthy manipulations of the equations which are underpinned by integrability
conditions inherent to the field equations.

\medskip

Most of the results concerning the conformal Einstein field equations available
in the literature make use of hyperbolic reductions leading to first order
\emph{symmetric hyperbolic evolution systems}. This approach works best for the
frame and spinorial versions of the conformal equations. Arguably, the simplest
variant of the conformal Einstein field equations is given by the so-called
\emph{metric conformal Einstein field equations} in which the field equations are
presented in tensorial form and the unphysical metric is determined by means of
an \emph{unphysical Einstein field equation} relating the Ricci tensor of the
unphysical metric to the various geometric fields entering in the conformal
equations --- these can be thought of as corresponding to some fictitious
unphysical matter. Remarkably, until recently, there was no suitable hyperbolic
reduction procedure available for this version of the conformal field
equations. In \cite{Pae15} Paetz has obtained a satisfactory hyperbolic
procedure for the metric vacuum Einstein field equations which is based on the
construction of second order wave equations. To round up his analysis, Paetz
then proceeds to construct a system of subsidiary wave equations for tensorial
fields encoding the conformal Einstein field equations (the so-called
\emph{geometric zero-quantities}) showing, in this way, the propagation of the
constraints. The motivation behind Paetz's approach is that the use of second
order hyperbolic equations gives access to a different part of the theory of
partial differential equations which complements the results available for
first order symmetric hyperbolic systems --- see e.g.
\cite{ChrPae13b,CarVal18b}. Paetz's construction of an evolution system
consisting of wave equations has been adapted to the case of the spinorial
conformal Einstein field equations in \cite{GasVal15}. In addition to its
interest in analytic considerations, the construction of wave equations for the
metric conformal Einstein field equations is also of relevance in numerical
studies, as the gauge fixing procedure and the particular form of the equations
is more amenable to implementation in current mainstream numerical codes than
other formulations of the conformal equations.

\medskip

The purpose of the present article is twofold: first, it generalises Paetz's
construction of a system of wave equations for the conformal Einstein field
equations to the case of matter models whose energy-momentum tensor has a
vanishing trace --- i.e. so-called \emph{tracefree matter}. The case of
tracefree matter is of particular interest since the equation of conservation
satisfied by the energy-momentum is conformally invariant; moreover, the
associated equations of motion for the matter fields can, usually, be shown to
possess good conformal properties --- see \cite{CFEBook}, Chapter 9. Second, it
clarifies the inner structure of Paetz's original construction by identifying
the integrability conditions underlying the mechanism of the propagation of the
constraints. The motivation behind this analysis is to extend the recent
analysis of the construction of vacuum anti-de Sitter-like spacetimes in
\cite{CarVal18b} to the case of tracefree matter.  However, we believe that the
analysis we present has an interest on its own right as it brings to the fore
the subtle structure of the metric conformal Einstein field equations.

\medskip
The main results of this article can be summarised as follows:

\begin{main}
The geometric fields in the metric conformal Einstein field equations
coupled to a tracefree matter field satisfy a system of wave
equations which is regular up to and beyond the conformal boundary of
a spacetime admitting a conformal extension. Moreover, the associated
geometric zero-quantities satisfy a (subsidiary) system of homogeneous wave
equations independently of the matter model. The subsidiary system is
also regular on the conformal boundary.
\end{main}

The precise statements concerning the above main result are contents
of Lemmas \ref{Lemma:EvolutionSystem} and \ref{Lemma:SubsidiarySystem}. 

\begin{remark}
{\em A remarkable property of our analysis is that it renders suitable
  evolution equations for the conformal fields and the zero-quantities without
  having to make any assumptions on the matter model except that it satisfies
  \emph{good} evolution equations in the conformally rescaled
  spacetime. Thus, our discussion can be regarded as a 
  \emph{once-for-all} analysis of the evolution equations associated
  to the geometric part of the metric conformal field equations valid for a
  wide class of coordinate gauges prescribed in terms of the
  coordinate gauge source function appearing in the \emph{generalised
  wave coordinate condition}.}
\end{remark}

\begin{remark}
{\em The homogeneity of the subsidiary system on the geometric
  zero-quantities is the key structural property required to ensure
  the propagation of the constraints by exploiting the uniqueness of
  solutions to a system of wave equations.}
\end{remark}

The approach followed to obtain our main result is based on the identification
of a family of integrability conditions associated to the metric conformal
Einstein field equations. To the best of our knowledge, these integrability
conditions have not appeared elsewhere in the literature. In our opinion this
approach brings better to the fore the structural properties of the conformal
Einstein field equations and, in particular, it makes the construction of the
subsidiary evolution system more transparent than the \emph{brute force}
approach adopted in \cite{Pae15}. A similar strategy is also adopted to study
the propagation of the gauge. In particular, by setting the matter fields to
zero, our analysis provides an alternative version of the main results of
\cite{Pae15} --- the initial conditions on the gauge required in the present
analysis differ from those in \cite{Pae15} though. Despite offering a more
sleek approach to the construction of an evolution system for the conformal
Einstein field equations, our analysis still requires heavy computations which
are best carried out in a computer algebra system. In the present case we have
made systematic use of the suite {\tt xAct} for the manipulation of tensorial
expressions in {\tt Mathematica} --- see \cite{xAct}. 

\medskip
We supplement our general analysis of the metric conformal Einstein
field equations with an analysis of the evolution and subsidiary
evolution equations of some of the tracefree matter models more
commonly used in the literature: the Maxwell field, the Yang-Mills
field and the conformally invariant scalar field. For each of these
fields we construct suitably second order wave equations for the
matter fields and the associated \emph{matter zero-quantities}. For
the case of the Yang-Mills field, our analysis makes no assumptions on
the gauge group. 

\medskip
As an application of our analysis, in the final section of this article
we present stability results for the de Sitter spacetime for
perturbations which include the Maxwell, Yang-Mills or conformally
invariant scalar field. Proofs of this result for the Maxwell and
Yang-Mills fields have been obtained in \cite{Fri91} using the spinorial
version of the conformal equations and a first order hyperbolic
reduction. The stability result for the conformally invariant scalar
field is, to the best of our knowledge, new.

\subsection*{Overview of the article}

In Section \ref{Section:MatterCFE} we briefly summarise the key properties of
the metric conformal Einstein field equations coupled to tracefree matter and
their relation to the Einstein field equations. Section \ref{Section:
WEGeomFields} provides the derivation of the geometric wave equations for the
geometric fields appearing in the conformal Einstein field equations. Section
\ref{Section: ZQandIC} introduces the key notion of geometric zero-quantity
and discusses the identities and integrability conditions associated to objects
of this type. Section \ref{Section:SubsidiarySystem} provides the construction
of the subsidiary evolution system for the geometric zero-quantities used in
the argument of the propagation of the constraints. This is, in principle, the
most calculationally intensive part of our analysis. However, using the
integrability conditions of Section \ref{Section: ZQandIC} we provide a
streamlined presentation thereof.  In Section \ref{Section:GaugeConsiderations}
we discuss the gauge freedom inherent in the geometric evolution systems
obtained in Sections \ref{Section: WEGeomFields} and
\ref{Section:SubsidiarySystem} and how this freedom can be used to complete the
hyperbolic reduction of the equations.  Section \ref{Section:PropOfGauge} 
establishes the consistency of the gauge introduced in the previous section,
independently of the particular tracefree matter model.  Section
\ref{Section:MatterFields} provides a case-by-case analysis of three
prototypical tracefree matter models --- the conformally invariant scalar field
(Subsection \ref{Subsection:ScalarField}), the Maxwell field (Subsection
\ref{Subsection:MaxwellField}) and the Yang-Mills field (Subsection
\ref{Subsection:YangMillsField}). The discussion for each of these matter
models includes the construction of suitable wave evolution equations and
subsidiary evolution equations. Finally, Section \ref{Section:Applications}
provides an application of the analysis developed in this article to the global
existence and stability of de Sitter-like spacetimes.

\subsection*{Conventions}
In what follows, $(\tilde{\mathcal{M}},\tilde{g}_{ab})$ will denote a spacetime
satisfying the Einstein equations with matter --- later we will make the
further assumption that the energy-momentum tensor is tracefree. The signature
of the spacetime metric is $(-,+,+,+)$. The lowercase Latin letters $a,\, b,\,
c, \ldots$ are used as abstract spacetime indices, while Greek letters $\mu, \,
\nu, \, \lambda,\ldots$ will be used as spacetime coordinate indices.  Our
conventions for the curvature are
\[
\nabla_c \nabla_d u^a -\nabla_d \nabla_c u^a = R^a{}_{bcd} u^b.
\]

\section{The metric conformal Einstein field equations with tracefree matter}
\label{Section:MatterCFE}
The purpose of this section is to provide a brief overview of the basic
properties of the conformal Einstein field equations with tracefree matter. A
more extended discussion of the properties of these equations, as well as their
derivation, can be found in Chapter 8 of \cite{CFEBook}. 

\subsection{Basic relations}
In what follows let $(\tilde{\mathcal{M}},\tilde{g}_{ab})$ denote a spacetime
satisfying the \emph{Einstein field equations with matter}
\begin{equation}
\tilde{R}_{ab} -\tfrac{1}{2}\tilde{R}\tilde{g}_{ab} + \lambda
\tilde{g}_{ab}= \tilde{T}_{ab},
\label{EFEMatter}
\end{equation}
where $\tilde{R}_{ab}$ and $\tilde{R}$ denote, respectively, the Ricci
tensor and Ricci scalar of the metric $\tilde{g}_{ab}$, $\lambda$ is
the Cosmological constant and $\tilde{T}_{ab}$ is the energy-momentum
tensor. As a consequence of the contracted Bianchi identity one
obtains the conservation law
\begin{equation}
\tilde{\nabla}^a \tilde{T}_{ab}=0.
\label{DivergencePhysicalEnergyMomentum}
\end{equation}
Here $\tilde{\nabla}_a$ denotes the Levi-Civita covariant derivative of the
metric $\tilde{g}_{ab}$. Now, let $(\mathcal{M},g_{ab})$ denote a spacetime
related to $(\tilde{\mathcal{M}},\tilde{g}_{ab})$ via a conformal embedding
\[
\tilde{\mathcal{M}}\stackrel{\varphi}{\hookrightarrow}\mathcal{M},
\qquad \tilde{g}_{ab} \stackrel{\varphi}{\mapsto} g_{ab}\equiv
\Xi^2 \big( \varphi^{-1})^*\tilde{g}_{ab}, \qquad
\Xi|_{\varphi(\tilde{\mathcal{M}})}>0,
\]
where is $\Xi$ a smooth scalar field --- the so-called \emph{conformal factor}.
With a slight abuse of notation we write
\begin{equation}
g_{ab} = \Xi^2 \tilde{g}_{ab}.
\label{Definition:ConformalRescaling}
\end{equation}
\begin{remark}
{\em Following the standard usage, we refer to
  $(\tilde{\mathcal{M}},\tilde{g}_{ab})$ as the \emph{physical
    spacetime} while  $(\mathcal{M},g_{ab})$ will be called the
  \emph{unphysical spacetime}. }
\end{remark}

\subsubsection{The unphysical energy-momentum tensor}
Since equation
\eqref{Definition:ConformalRescaling} does not determine the way
$\tilde{T}_{ab}$ transforms, it will be convenient to define the
\emph{unphysical energy-momentum tensor} as
\[
T_{ab}\equiv \Xi^{-2} \tilde{T}_{ab}.
\]
Using the transformation rules between the Levi-Civita covariant derivatives of
conformally related metrics, equation \eqref{DivergencePhysicalEnergyMomentum}
takes the form
\[
\nabla^a T_{ab} = \Xi^{-1} T \nabla_b \Xi,
\]
with $\nabla_a$ the Levi-Civita covariant derivative of $g_{ab}$
and $T\equiv g^{ab} T_{ab}$. It then follows that
\[
\nabla^a T_{ab}=0 \qquad \mbox{ if and only if} \qquad T=0.
\]

\begin{assumption}
\label{Assumption:AssumptionDivergenceUnphysicalEnergyMomentum}
{\em In the remainder of this article we restrict our attention to matter
models for which $T=0$, so that the corresponding unphysical energy-momentum
tensor $T_{ab}$ is divergence-free, that is,
\begin{equation}
\nabla^a T_{ab}=0.
\label{DivergenceUnphysicalEnergyMomentum}
\end{equation}}
\end{assumption}

\subsection{Basic properties of the conformal Einstein field
  equations}
\label{Section:BasicPropertiesCFE}

The \emph{metric tracefree conformal Einstein field equations} have been first
discussed in \cite{Fri91}. In terms of the notation and conventions used in
this article they are given by
\begin{subequations}
\begin{eqnarray}
&&\nabla_{a}\nabla_{b}\Xi = -\Xi
L_{ab}+sg_{ab}+\tfrac{1}{2}\Xi^{3}T_{ab}, \label{TraceCFE1}
\\
&&\nabla_{a}s=-L_{ab}\nabla^{b}\Xi+\tfrac{1}{2}\Xi^{2}\nabla^{b}\Xi
T_{ab}, \label{TraceCFE2}
\\
&&\nabla_{a}L_{bc}-\nabla_{b}L_{ac}=\nabla_{e}\Xi
d^{e}{}_{cab}+\Xi T_{abc}, \label{TraceCFE3}
\\
&&\nabla_{e}d^{e}{}_{abc}=T_{bca}, \label{TraceCFE4}
\\
&&6\Xi s-3\nabla_{c}\Xi\nabla^{c}\Xi=\lambda, \label{TraceCFE5}\\
&& R^c{}_{dab} = \Xi d^c{}_{dab} + 2( \delta^c{}_{[a}L_{b]d} -
   g_{d[a}L_{b]}{}^c). \label{TraceCFE6}
\end{eqnarray}
\end{subequations}
A detailed derivation of these equations can be found in \cite{CFEBook}.  In
the above expressions $L_{ab}$, $s$, $d^a{}_{bcd}$ and $T_{abc}$ denote,
respectively, the Schouten tensor, the Friedrich scalar, the rescaled Weyl
tensor and the rescaled Cotton tensor. These objects are defined as
\begin{subequations}
\begin{eqnarray}
&& L_{ab}\equiv\tfrac{1}{2}R_{ab}-\tfrac{1}{12}g_{ab}R, \label{SchoutenTensor} \\
&& s\equiv \tfrac{1}{4}\nabla^{c}\nabla_{c}\Xi+\tfrac{1}{24}R\Xi, \label{FriedrichScalarDefinition}
\\
&& d^a{}_{bcd} \equiv \Xi^{-1} C^a{}_{bcd},\\
&& T_{abc} \equiv \Xi \nabla_{[a}T_{b]c} +3 \nabla_{[a} \Xi
T_{b]c} - g_{c[a}T_{b]e} \nabla^e\Xi, \label{RescaledCottonTensor}
\end{eqnarray}
\end{subequations}
where $C^a{}_{bcd}$ is the conformally invariant Weyl tensor. Observe that
$T_{abc}$ has the following symmetries:
\begin{equation}
T_{abc} = T_{[ab]c}, \quad T_{[abc]} = 0.
\label{CottonTensorProperties}
\end{equation}

Relevant for the subsequent discussion is the well-known fact that the
rescaled Weyl tensor has two associated Hodge dual tensors, namely
\[
\,^*d_{abcd} \equiv \tfrac{1}{2} \epsilon_{ab}{}^{ef} d_{efcd}, \quad
d^*_{abcd} \equiv \tfrac{1}{2} \epsilon_{cd}{}^{ef} d_{abef},
\]
where $\epsilon_{abcd}$ is the 4-volume form of the metric $g_{ab}$. One can
check that $\,^*d_{abcd} = d^*_{abcd}$. Similarly, we also define the Hodge
dual of $T_{abc}$ as
\begin{equation}
\,^*T_{abc} \equiv \tfrac{1}{2} \epsilon_{ab}{}^{de} T_{dec}.
\label{DefinitionCottonDual}
\end{equation}
Moreover, if \Cref{Assumption:AssumptionDivergenceUnphysicalEnergyMomentum} and
equation \eqref{TraceCFE1} are taken into account, one obtains some additional
relations, namely
\begin{subequations}
\begin{eqnarray}
&& \nabla_c T_{ab}{}^c  = 0, \\
&& \nabla_c \,^*T_{ab}{}^c = 0, \\
&& \nabla_c T_a{}^c{}_b = \nabla_c T_{(a}{}^c{}_{b)}.
\end{eqnarray}
\end{subequations}

\begin{remark}
{\em Equations \eqref{TraceCFE1}-\eqref{TraceCFE4} will be regarded as
a set of differential conditions for the fields $\Xi$, $s$, $L_{ab}$
and $d^a{}_{bcd}$. Equation \eqref{TraceCFE5} can be shown to play the
role of a constraint which only needs to be verified at a single point
--- see e.g. \cite{CFEBook}, Lemma 8.1. Equation \eqref{TraceCFE6},
providing the link between the conformal fields $d^c{}_{dab}$,
$L_{ab}$ and the irreducible decomposition of the Riemann tensor, 
allows us to deduce a differential condition for the components of the
unphysical metric $g_{ab}$ --- see Section 
\ref{Section:GeneralisedWaveCoordinates}}.
\end{remark}

\begin{remark}
{\em By a solution to the metric tracefree conformal Einstein field equations
it will be understood a collection of fields
$(g_{ab},\Xi,s,L_{ab},d^a{}_{bcd},T_{ab})$ satisfying equations
\eqref{DivergenceUnphysicalEnergyMomentum} and
\eqref{TraceCFE1}-\eqref{TraceCFE6}.}
\end{remark}

\medskip
The relation between the metric tracefree conformal Einstein field equations
and the Einstein field equations \eqref{EFEMatter} is given in the following
proposition --- see \cite{CFEBook}, Proposition 8.1. 

\begin{proposition}
Let $(g_{ab},\Xi,s,L_{ab},d^a{}_{bcd},T_{ab})$ denote a solution to the
metric tracefree conformal Einstein field equations such that $\Xi\neq 0$ on an open
set $\mathcal{U}\subset \mathcal{M}$. Then the metric $\tilde{g}_{ab} =\Xi^{-2}
g_{ab}$ is a solution to the Einstein field equations \eqref{EFEMatter} with
energy momentum tensor given by $\tilde{T}_{ab} = \Xi^2 T_{ab}$ on
$\mathcal{U}$. 
\end{proposition}

\begin{proof}
The proof given in \cite{CFEBook} omits equation \eqref{TraceCFE6} and
implicitly assumes that the field $L_{ab}$ can be identified with the
Schouten tensor of the metric $g_{ab}$. With equation  \eqref{TraceCFE6}
at hand, one is allowed to make this identification. From here
onwards one can apply the argument in Proposition 8.1 in \cite{CFEBook}.
\end{proof}

\subsubsection{An alternative equation for $d^a{}_{bcd}$}

For our purposes, it will be convenient to consider an alternative version of
the conformal field equation for the rescaled Weyl tensor. This can be obtained
as follows: multiplying \eqref{TraceCFE4} by $\epsilon_{fg}{}^{bc}$ and
exploiting the identity $\,^*d_{abcd} = d^*_{abcd}$ results in
\[
2\nabla_a{}^*d_{fgc}{}^a = 2\nabla_a{}d^*_{fgc}{}^a = -2 {}^*T_{fgc}.
\]
From here it follows that
\begin{equation}
3 \nabla_{[e}d_{ab]cd} + \epsilon_{eabf}{}^*T_{cd}{}^f = 0. \label{TraceCFE4Alt}
\end{equation}

\begin{remark}
\label{Remark: TraceCFE4Alt}
\em {This last equation is equivalent to \eqref{TraceCFE4} and will be
essential in sections \ref{Section: WEGeomFields} and \ref{Section: ZQandIC}
where a system of wave equations for the geometric fields and the zero-quantities
associated to the equations \eqref{TraceCFE1}-\eqref{TraceCFE6} is discussed.}
\end{remark}

\subsection{An equation for the components of the metric $g_{ab}$}

Taking the natural trace in equation \eqref{TraceCFE6} leads to the relation
\begin{equation}
R_{ab} = 2L_{ab} + \tfrac16 R g_{ab}. \label{EquationMetric}
\end{equation}
Here, $R_{ab}$ and $L_{ab}$ are considered as independent
variables. In particular, the Ricci tensor $R_{ab}$ is assumed to be
expressed in terms of first and second derivatives of the components of the
metric whilst $L_{ab}$ is a field satisfying equations
\eqref{TraceCFE1}-\eqref{TraceCFE5}. This will be further discussed in
\Cref{Section:GaugeConsiderations} where a suitable wave equation for the
components of the metric is constructed. 

\begin{remark}
{\em As pointed out in \cite{Fri03a}, equation \eqref{EquationMetric}
  can be regarded as an Einstein field equation for the unphysical
  metric $g_{ab}$. From this point of view, the geometric fields $\Xi$,
$s$, $L_{ab}$ and $d_{abcd}$ can be regarded as unphysical matter
fields. Accordingly, in the following we refer to
equation \eqref{EquationMetric} as the \emph{unphysical Einstein
  equation}. This approach should allow to adapt well-tested
numerical methods for the Einstein field equations to the case of the
conformal field equations.}
\end{remark}

\section{The evolution system for the geometric fields}
\label{Section: WEGeomFields}
In this section we show how to construct an evolution system for the
geometric fields appearing in the conformal Einstein field equations,
equations \eqref{TraceCFE1}-\eqref{TraceCFE6}. These evolution
equations take the form of \emph{geometric wave equations}
--- that is, their principal part involves the D'Alambertian $\square
\equiv \nabla_a \nabla^a$ associated to the conformal metric $g_{ab}$. 

\medskip
In \cite{Pae15}, Paetz has obtained a system of geometric wave equations for
the set of conformal fields $(\Xi, s, L_{ab}, d^a{}_{bcd})$ in the vacuum case.
This can be generalised to include a tracefree matter component. The next statement
summarises this result:
\begin{lemma}
\label{Lemma:EvolutionSystem}
The metric tracefree conformal Einstein field equations
\eqref{TraceCFE1}-\eqref{TraceCFE6} imply the following system of geometric
wave equations for the conformal fields:
\begin{subequations}
\begin{eqnarray}
&& \square\Xi = 4s - \tfrac16 \Xi R, \label{CWE1} \\
&& \square s = - \tfrac{1}{6} s R + \Xi L_{ab} L^{ab} - \tfrac{1}{6} \nabla_{a}R \nabla^{a}\Xi
+ \tfrac{1}{4} \Xi^5 T_{ab} T^{ab} - \Xi^3 L_{ab}T^{ab} 
+ \Xi \nabla^{a}\Xi \nabla^{b}\Xi T_{ab}, \label{CWE2} \\
&& \square L_{ab} = -2 \Xi d_{acbd} L^{cd} + 4 L_{a}{}^{c} L_{bc} - L_{cd} L^{cd} g_{ab} 
+ \tfrac{1}{6} \nabla_{a}\nabla_{b}R + \tfrac{1}{2} \Xi^3 d_{acbd} T^{cd} 
\nonumber \\
&& \hspace{1.3cm} - \Xi \nabla_{c}T_{a}{}^{c}{}_{b} - 2 T_{(a|c|b)}
   \nabla^{c}\Xi, \label{CWE3} \\
&& \square d_{abcd} = - 4 \Xi d_{a}{}^{f}{}_{[c}{}^{e} d_{d]ebf} - 2 \Xi d_{a}{}^{f}{}_{b}{}^{e} d_{cdfe} 
+ \tfrac{1}{2} d_{abcd} R - T_{[a}{}^f \Xi^2 d_{b]fcd} - \Xi^2 T_{[c}{}^f d_{d]fab} \nonumber \\
&& \hspace{1.5cm} - \Xi^2 g_{a[c} d_{d]gbf} T^{fg} + \Xi^2 g_{b[c}  d_{d]gaf} T^{fg} 
+ 2 \nabla_{[a}T_{|cd|b]} + \epsilon_{abef} \nabla^{f}\,^*T_{cd}{}^{e}. \label{CWE4Alternative}
\end{eqnarray}
\end{subequations}
\end{lemma}

\begin{proof}
Equation \eqref{CWE1} is a direct consequence of \eqref{TraceCFE1}. Equations
\eqref{CWE2} and \eqref{CWE3} result, respectively, from applying a covariant
derivative to \eqref{TraceCFE2} and \eqref{TraceCFE3}, and using the second
Bianchi identity.  The wave equation for $d^a{}_{bcd}$, on the other hand,
requires to consider the alternative conformal field equation
\eqref{TraceCFE4Alt}.  Applying $\nabla^e$ to the latter and using equation
\eqref{TraceCFE4} along with the first Bianchi identity, a long but
straightforward calculation yields the wave equation 
\begin{eqnarray}
&& \square d_{abcd} = - 4 \Xi d_{a}{}^{f}{}_{[c}{}^{e} d_{d]ebf} - 2 \Xi d_{a}{}^{f}{}_{b}{}^{e} d_{cdfe} 
+ \tfrac{1}{3} d_{abcd} R - 2d_{cdf[a} L_{b]}{}^{f} - 2d_{abf[c} L_{d]}{}^{f} 
 \nonumber \\
 && \hspace{1.5cm}- 2g_{a[c} d_{d]ebf} L^{fe} + 2 g_{b[c} d_{d]fae} L^{ef} 
 + 2\nabla_{[a}T_{|cd|b]} + \epsilon_{abef}\nabla^f {}^*T_{cd}{}^e. \label{CWE4}
\end{eqnarray}
It is possible to eliminate terms containing $L_{ab}$ from the wave
equation \eqref{CWE4} through the generalisation of an identity obtained in \cite{Pae15} to
the case of tracefree matter.  Multiplying equation \eqref{TraceCFE4Alt} by $\Xi$, using
the definitions of $d^a{}_{bcd}$ and $^*T_{abc}$, equation \eqref{TraceCFE3}
and the second Bianchi identity to simplify it, one finds that
\begin{equation}
d_{cd[ag}\nabla_{b]}\Xi + d_{de[ag}g_{b]c}\nabla^e\Xi - d_{ce[ag}g_{b]d}\nabla^e\Xi = 0.
\label{PaetzIdentityInter}
\end{equation}
Applying a further covariant derivative $\nabla^g$ to the last expression and
making use of equations \eqref{TraceCFE1}, \eqref{TraceCFE4} and
\eqref{TraceCFE4Alt} as well as the properties of the rescaled Cotton tensor,
the following identity is obtained:
\begin{eqnarray}
&& 2\Xi d_{cdf[a} L_{b]}{}^{f} + 2\Xi d_{abf[c} L_{d]}{}^{f} 
+ 2 g_{a[c} \Xi d_{d]gbf} L^{fg} - 2\Xi g_{b[c} d_{d]gaf} L^{fg}  
+ \tfrac{1}{6} \Xi d_{abcd} R \nonumber \\
&& - \Xi^3 d_{cdf[a} T_{b]}{}^{f} - \Xi^3 d_{abf[c} T_{d]}{}^{f} 
- \Xi^3 g_{a[c} d_{d]gbf} T^{fg} + \Xi^3 g_{b[c} d_{d]gaf} T^{fg} = 0.
\label{PaetzIdentity}
\end{eqnarray}
By substituting this into expression \eqref{CWE4} we get equation
\eqref{CWE4Alternative}, which does not involve the Schouten tensor.
\end{proof}

\begin{remark}
{\em In concrete applications it may prove useful to express the Schouten
tensor in terms of the tracefree Ricci tensor and the Ricci scalar through the
formula
\begin{equation}
L_{ab} = \Phi_{ab} + \tfrac{1}{24}R g_{ab}.
\label{DecompositionSchouten}
\end{equation}
As will be discussed in Section \ref{Section:ConformalGaugeSourceFunctions},
the Ricci scalar $R$ is associated to the particular choice of conformal gauge.
Thus, the decomposition \eqref{DecompositionSchouten} allows us to split the
field $L_{ab}$ into a gauge part and a part which is determined through the
field equations. Keeping the simplicity of the presentation in mind, we do not
pursue this approach further as it leads to lengthier expressions.}
\end{remark}

\section{Zero-quantities and integrability conditions}
\label{Section: ZQandIC}
In this section we consider a convenient setting for the discussion and
book-keeping of the evolution equations implied by the conformal Einstein field
equations with tracefree matter. Our approach is based on the observation that
the metric conformal Einstein field equations constitute an overdetermined
system of differential conditions for the various conformal fields. Thus, the
equations are related to each other through \emph{integrability conditions} ---
i.e.  necessary conditions for the existence of solutions to the equations. 

\subsection{Definitions and basic properties}

First we proceed to introduce the set of \emph{geometric zero-quantities} (also
called \emph{subsidiary variables}) associated to the system of metric
tracefree conformal Einstein field equations
\eqref{TraceCFE1}-\eqref{TraceCFE5}. These fields are defined as:
\begin{subequations}
\begin{eqnarray}
&& \Upsilon_{ab} \equiv \nabla_{a}\nabla_{b}\Xi +\Xi
L_{ab}+sg_{ab} - \tfrac{1}{2}\Xi^{3}T_{ab}, \label{ZQ1} \\
&& \Theta_a \equiv \nabla_{a}s + L_{ac}\nabla^{c}\Xi - \tfrac{1}{2}\Xi^{2}\nabla^{c}\Xi
T_{ac}, \label{ZQ2} \\
&& \Delta_{abc} \equiv  \nabla_{a}L_{bc}-\nabla_{b}L_{ac} - \nabla_{e}\Xi d^{e}{}_{cab} 
- \Xi T_{abc}, \label{ZQ3} \\
&& \Lambda_{abc} \equiv T_{bca} - \nabla_{e}d^{e}{}_{abc}, \label{ZQ4} \\
&& Z \equiv \lambda - 6\Xi s
   +3\nabla_{c}\Xi\nabla^{c}\Xi, \label{ZQ5}\\
&& P^c{}_{dab}\equiv  R^c{}_{dab} - \Xi d^c{}_{dab} - 2( \delta^c{}_{[a}L_{b]d} -
   g_{d[a}L_{b]}{}^c). \label{ZQ6}
\end{eqnarray}
\end{subequations}
In terms of the above, the conformal Einstein field equations
\eqref{TraceCFE1}-\eqref{TraceCFE6} can be expressed as the conditions
\[
\Upsilon_{ab}=0, \qquad  \Theta_a=0, \qquad  \Delta_{abc} = 0, \qquad
\Lambda_{abc}=0, \qquad Z=0,\qquad P^c{}_{dab}=0,
\]
from where these fields take their name.

\subsubsection{Properties of the zero-quantities}
By definition, the zero-quantities possess the following symmetries:
\begin{equation}
\begin{gathered}
\Upsilon_{ab} = \Upsilon_{(ab)}, \quad \Delta_{abc} = \Delta_{[ab]c}, \quad
\Delta_{[abc]} = 0, \quad \Lambda_{abc} = \Lambda_{a[bc]}, \quad \Lambda_{[abc]} =0, \\
\Delta_a{}^b{}_b = 0, \quad \Lambda^b{}_{ab} = 0. 
\label{ZQIdentities}
\end{gathered}
\end{equation}
Moreover, one can check that $\Delta_{abc}$ and $\Lambda_{abc}$ satisfy the
identities
\begin{equation}
\Delta_{abc} = \tfrac23 \Delta_{abc} + \tfrac13\Delta_{acb} - \tfrac13\Delta_{bca}, \qquad
\Lambda_{abc} = \tfrac23 \Lambda_{abc} + \tfrac13\Lambda_{bac} - \tfrac13\Lambda_{cab}, \label{LanczosDeltaLambda}
\end{equation}
which are useful for simplifying certain combinations of zero-quantities. Regarding
$P^a{}_{bcd}$, it inherits the symmetries of the Riemann tensor; in particular,
we can define its Hodge dual tensors
\begin{equation}
^*P_{abcd} \equiv \tfrac12\epsilon_{ab}{}^{ef} P_{efcd},
\quad P^*_{abcd} \equiv \tfrac12\epsilon_{cd}{}^{ef} P_{abef}.
\label{ZQ6Duals}
\end{equation}

In addition, it will result useful to introduce a further auxiliary
zero-quantity associated to equation \eqref{TraceCFE4Alt} --- see
\Cref{Remark: TraceCFE4Alt}:
\begin{eqnarray}
&& \Lambda_{abcde} \equiv 3\nabla_{[a}d_{bc]de} + \epsilon_{abcf}{}^*T_{de}{}^f 
= 3\Lambda_{d[ab} g_{c]e}  - 3\Lambda_{e[ab} g_{c]d}.
\label{Lambda_abcde}
\end{eqnarray}
Here, the second equality has been obtained through a calculation similar to
the one yielding \eqref{TraceCFE4Alt}. From the above definition it follows that
$\Lambda_{ab}{}^d{}_{cd} = \Lambda_{abc}$, as well as
\begin{equation}
\Lambda_{abcde} = \Lambda_{[abc]de}, \qquad \Lambda_{abcde} =
\Lambda_{abc[de]}.
\label{Identity:LambdaFiveIndices}
\end{equation}

\subsubsection{Some consequences of the wave equations}

Key for our subsequent analysis is the observation that assuming the validity
of the geometric wave equations for the conformal fields implies a further set of
relations satisfied by the zero-quantities. These are summarised in the
following lemma:

\begin{lemma}
\label{LemmaZQWE}
Assume that the wave equations \eqref{EquationMetric},
\eqref{CWE1}-\eqref{CWE4Alternative}, and
\Cref{Assumption:AssumptionDivergenceUnphysicalEnergyMomentum} hold. Then the
geometric zero-quantities satisfy the identities
\begin{subequations}
\begin{eqnarray}
&& \Upsilon_a{}^a = 0, \label{ZQWE1} \\
&& P^c{}_{acb} = 0, \label{ZQWE6} \\
&& \nabla_b \Upsilon_a{}^b = 3 \Theta_a, \label{DivergenceZQ1} \\
&& \nabla_a\Theta^a = \Upsilon^{ab}L_{ab} -\tfrac12 \Xi^2 \Upsilon^{ab}T_{ab}, \label{ZQWE2} \\
&& \nabla_c\Delta_a{}^c{}_b = \Upsilon^{cd}d_{acbd} + \Lambda_{abc}\nabla^c\Xi
- L^{cd}P_{acbd}, \label{ZQWE3} \\
&& \nabla_c\Delta_{ab}{}^c  = 2\Xi T_{c[a}\Upsilon_{b]}{}^c -\Lambda_{cab}\nabla^c\Xi, \label{DivergenceDelta1} \\
&& \nabla_c \Lambda^c{}_{ab} = d_{[a}{}^{cde}P_{b]cde} -2T_{c[a}\Upsilon_{b]}{}^c, \label{DivergenceLambda1} \\
&& \nabla_c \Lambda_{[ab]}{}^c = 2d_{[a}{}^{cde}P_{b]dec}, \label{SymmetricDivergenceLambda} \\
&& \nabla_d P_{abc}{}^d = -\Delta_{abc} -\Xi \Lambda_{cab} , \label{DeltaLambdaIdentity} \\
&& \nabla_{c}\Lambda_{eg}{}^{c}{}_{mn} = 2 \nabla_{[e}\Lambda_{g]mn} + 2
d_{[e}{}^{c}{}_{|m|}{}^{h}P_{g]cnh} - 2 d_{[e}{}^{c}{}_{|n|}{}^{h}P_{g]cmh} + 2
d_{mn}{}^{ch}P_{ecgh}. \label{ZQWE4}
\end{eqnarray}
\end{subequations}
\end{lemma}

\begin{proof}
The result follows directly from the definitions of the zero-quantities with
the aid of the wave equations for the conformal fields \eqref{EquationMetric}
and \eqref{CWE1}-\eqref{CWE4Alternative}, the second Bianchi identity and the
properties of the rescaled Cotton tensor. It is worth mentioning that
\eqref{ZQWE4} is obtained by using \eqref{CWE4} instead of
\eqref{CWE4Alternative} as it considerably simplifies the calculation.
\end{proof}

\subsection{Integrability conditions}

The zero-quantities are not independent of each other but they are related via
a set of identities, the so-called \emph{integrability conditions}.  These
relations are key for the computation of a suitable (subsidiary) system of wave
equations for the zero-quantities. The procedure to obtain these relations is
to compute suitable antisymmetrised covariant derivatives of the
zero-quantities which, in turn, are expressed in terms of lower order objects.
Following this general strategy we obtain the following:

\begin{proposition}
\label{Proposition:IntegrabilityConditions}
The geometric zero-quantities defined in \eqref{ZQ1}-\eqref{ZQ3} and
\eqref{ZQ5}-\eqref{ZQ6} satisfy the identities
\begin{subequations}
\begin{eqnarray}
&& 2\nabla_{[a}\Upsilon_{c]b}  = 2g_{b[a} \Theta_{c]} + \Xi
   \Delta_{acb} + P_{acbd}\nabla^d\Xi, \label{IC1}\\
&& 2\nabla_{[a}\Theta_{b]}  = -  2L_{[a}{}^{c} \Upsilon_{b]c} +
   \Delta_{abc} \nabla^{c}\Xi + \Xi^2
   T_{c[a}\Upsilon_{b]}{}^c, \label{IC2}\\
&& 3\nabla_{[d}\Delta_{ab]c} =
\Lambda_{abdce} \nabla^{e}\Xi + 3 \Upsilon_{[a}{}^{e}d_{bd]ce} + 3
L_{[a}{}^{e}P_{bd]ce} -  \tfrac{3}{2} \Xi^2 P_{[ab|c|}{}^{e}T_{d]e} + 2
\Xi \Upsilon_{[a}{}^{e}g_{b|c|}T_{d]e} \nonumber \\ 
&& \hspace{2cm} + \Xi \Upsilon_{[a}{}^{e}g_{|c|b}T_{d]e}, \label{IC3} \\
&& \nabla_a Z = -6\Xi \Theta_a + 6 \Upsilon_{ab}\nabla^b\Xi, \label{IC5} \\
&& 3 \nabla_{[e}P_{gh]mn} = \Xi \Lambda_{eghnm} - 3 \Delta_{[eg|m|}g_{h]n} + 3 \Delta_{[eg|n|}g_{h]m}.
\label{IC6}
\end{eqnarray}
\end{subequations}
\end{proposition}

\begin{proof}
Equations \eqref{IC1}-\eqref{IC5} follow from direct calculations employing the
definitions of the zero-quantities, the rescaled Cotton tensor and the first
Bianchi identity. Equation \eqref{IC6}, on the other hand, can be obtained in a
similar manner as \eqref{TraceCFE4Alt}: multiplying \eqref{DeltaLambdaIdentity}
by $\epsilon_{mn}{}^{cd}$ and exploiting the fact that $^*P_{abcd} =
P^*_{abcd}$ --- which is a consequence of \eqref{ZQWE6} --- yields
\begin{equation}
2\nabla_a{}^*P_{mnb}{}^a = 2 \nabla_a P^*_{mnb}{}^a =  -\epsilon_{mnac}(\Xi\Lambda_b{}^{ac} + \Delta^{ac}{}_b).
\end{equation}
By substituting back the definition of $P^*_{mnab}$, \eqref{IC6} is found after
some simplifications.
\end{proof}

\begin{remark}
{\em  Observe that these relations have right-hand sides consisting of lower
order expressions which are homogeneous in the zero-quantities. This property
will be key when suitable wave equations for these fields are derived in the
next section. Equations \eqref{IC1}-\eqref{IC6} together with \eqref{ZQWE4}
constitute the set of integrability conditions for the geometric zero-quantities
associated to the tracefree conformal Einstein field equations.}
\end{remark}

\begin{remark}
\label{Remark:EquivalenceWaveEquationsWeyl}
{\em The expressions in Lemma \eqref{LemmaZQWE} and Proposition
  \eqref{Proposition:IntegrabilityConditions} allow us to show, in particular, that
  the wave equations \eqref{CWE4Alternative} and \eqref{CWE4} differ
  from each other by a homogeneous combination of
  zero-quantities. Thus, in arguments involving the propagation of
  the constraints, both forms of the evolution equation can be used interchangeably.}
\end{remark}

\section{The subsidiary evolution system for the zero-quantities}
\label{Section:SubsidiarySystem}

An important aspect of any \emph{hyperbolic reduction procedure} for the
(conformal) Einstein field equations is the identification of the conditions
upon which a solution to the (reduced) evolution equations implies a
solution to the full set of field equations --- this type of analysis is
generically known as the \emph{propagation of the constraints}. In
practice, the propagation of the constraints requires the construction
of a suitable system of evolution equations for the zero-quantities
associated to the field equations. 

\subsection{Construction of the subsidiary system}

In this section it is shown how the set of integrability conditions
provides a systematic and direct way to obtain wave
equations for the zero-quantities --- a so-called \emph{subsidiary evolution
system}. The propagation of the constraints then follows from the structural
properties of the subsidiary system as a consequence of the uniqueness
of solutions to systems of wave equations.  

\subsubsection{Equations for $\Upsilon_{ab}$, $\Theta_a$,
  $\Delta_{abc}$, $Z$ and $P_{abcd}$}

Equation \eqref{IC1} serves as the starting point to obtain a wave equation for
$\Upsilon_{ab}$. After applying $\nabla^c$ and commuting derivatives, equation
\eqref{DivergenceZQ1} renders it as a suitable wave equation. Remaining first
order derivatives can be rewritten and simplified via equations
\eqref{LanczosDeltaLambda}, \eqref{ZQWE1}, \eqref{ZQWE2}, \eqref{ZQWE3} and
\eqref{DeltaLambdaIdentity} resulting in:
\begin{align}
\square \Upsilon_{ab} = & \tfrac{1}{6} \Upsilon_{ab} R - 2 \Upsilon^{cd}
L_{cd} g_{ab} + \tfrac{1}{2} \Xi^2 \Upsilon^{cd} g_{ab} T_{cd} + 4
\nabla_{(a}\Upsilon_{b)} - 2 \Xi \Upsilon^{cd}d_{acbd} + 4\Upsilon_{(a}{}^{c}L_{b)c} \nonumber \\
& - 2 \Upsilon^{cd}P_{acbd} + 2 \Xi
L^{cd}P_{acbd} - \tfrac{1}{2} \Xi^3 P_{abcd}T^{cd}.
\label{WEZQ1}
\end{align}

Regarding $\Theta_a$, an analogous calculation using expression \eqref{IC2} in
conjunction with the same equations as in the previous case leads directly to a
wave equation for this field. Exploiting \eqref{TraceCFE3},
\eqref{RescaledCottonTensor} and \eqref{IC1} to simplify it one obtains
\begin{align}
\square\Theta_c = & \ 6 L_{ca} \Theta^{a} - 2 \Upsilon^{ab}
\Delta_{cab} + 2 \Xi L^{ab} \Delta_{cab} -  \Xi^3 \Delta_{c}{}^{ab}
T_{ab} - 2 \Xi^2 \Theta^{a} T_{ca} - 2 \Upsilon^{bd} d_{cbad} \nabla^{a}\Xi \nonumber \\
& + \tfrac{3}{2} \Xi \Upsilon_{c}{}^{b} T_{ab} \nabla^{a}\Xi + \tfrac{1}{2} \Xi^2
P_{cbad} T^{bd} \nabla^{a}\Xi + \tfrac{1}{2} \Xi \Upsilon_{a}{}^{b} T_{cb}
\nabla^{a}\Xi - \tfrac{1}{6} \Upsilon_{ca} \nabla^{a}R -  \tfrac{5}{2} \Xi
\Upsilon^{ab} T_{ab} \nabla_{c}\Xi \nonumber \\
& + 2 \Upsilon^{ab} \nabla_{c}L_{ab} -  \Xi^2 \Upsilon^{ab}\nabla_{c}T_{ab}.
\label{WEZQ2}
\end{align}

A wave equation for $\Lambda_{abc}$ can be obtained by applying $\nabla^d$ to
integrability condition \eqref{IC3}, commuting derivatives and using
\eqref{ZQWE3} to eliminate the second order derivatives. A direct but long
calculation exploiting the same relations used in the previous two cases, along
with \eqref{TraceCFE4} and \eqref{Lambda_abcde}, yields
\begin{align}
\square \Delta_{abc} = & \ 2 \Lambda_{cab} s - \Upsilon_{c}{}^{d} T_{abd} -  \Xi \Lambda_{abdce} L^{de} + 3
d_{abcd} \Theta^{d} + \tfrac{1}{3} R \Delta_{abc} + L_{c}{}^{d} \Delta_{abd} +
\tfrac{1}{2} \Xi^3 \Lambda_{abdce} T^{de} \nonumber \\
& - \Xi P_{abce} T_{d}{}^{e} \nabla^{d}\Xi + \tfrac{1}{6} P_{abcd} \nabla^{d}R + \nabla^{d}\Xi
\nabla_{e}\Lambda_{ab}{}^{e}{}_{cd} + 2 \Upsilon^{de} \nabla_{e}d_{abcd} + L^{de} \nabla_{e}P_{abcd} \nonumber \\
&  - \tfrac{1}{2} \Xi^2 T^{de} \nabla_{e}P_{abcd} + 2 \Upsilon_{[a}{}^{d}T_{b]cd} 
- \Xi \Upsilon_{[a}{}^{d}\nabla_{|c|}T_{b]d} 
- 2 \Xi d_{[a}{}^{d}{}_{b]}{}^{e}\Delta_{dec} + 2 \Xi d_{[a}{}^{d}{}_{|c|}{}^{e}\Delta_{b]de} \nonumber \\
& + 2 d_{[a}{}^{d}{}_{|c|}{}^{e}\nabla_{b]}\Upsilon_{de} - 2d_{[a}{}^{d}{}_{|c}{}^{e}\nabla_{d|}\Upsilon_{b]e} 
- 2 L_{[a}{}^{d}\Delta_{b]dc} + 2L^{de}\nabla_{[a}P_{b]dce} - 2 P_{[a}{}^{d}{}_{b]}{}^{e}\Delta_{dec} \nonumber \\ 
& + 2 P_{[a}{}^{d}{}_{|c|}{}^{e}\Delta_{b]de} - 2 P_{[a}{}^{d}{}_{|c|}{}^{e}\nabla_{b]}L_{de} 
- 2 P_{[a}{}^{d}{}_{|c}{}^{e}\nabla_{d|}L_{b]e} + \Xi^2 P_{[a}{}^{d}{}_{|c}{}^{e}\nabla_{d|}T_{b]e} 
- \Xi^2 \Delta_{c}{}^{d}{}_{[a}T_{b]d} \nonumber \\
&  + \Xi T_{[a}{}^{d}\nabla_{|c|}\Upsilon_{b]d} - 2 \nabla^{d}\Xi\nabla_{[a}\Lambda_{b]cd} 
+ 2 \Upsilon^{de}T_{[a|de|}g_{b]c} 
+ \Xi \Upsilon^{de}g_{[a|c}\nabla_{d|}T_{b]e} - \Upsilon_{[a}{}^{d}T_{b]d}\nabla_{c}\Xi \nonumber \\
& - 2 L^{de}\Delta_{[a|de|}g_{b]c} + 3 \Xi \Upsilon^{d}g_{[a|c|}T_{b]d} 
+ 2 \Xi P_{[a}{}^{d}{}_{|c|}{}^{e}T_{b]e}\nabla_{d}\Xi 
- \Xi g_{[a|c}T^{de}\nabla_{d|}\Upsilon_{b]e} \nonumber \\
& + \Upsilon_{[a}{}^{d}g_{b]c}T_{d}{}^{e}\nabla_{e}\Xi + \Upsilon^{de}g_{[a|c|}T_{b]d}\nabla_{e}\Xi.
\label{WEZQ3}
\end{align}

A wave equation for $Z$ is readily found by simply applying $\nabla^a$ to
equation \eqref{IC5}:
\begin{equation}
\square Z = 6 \Upsilon_{ab} \Upsilon^{ab} - 12 \Xi \Upsilon^{ab} L_{ab} 
+ 6 \Xi^3 \Upsilon^{ab} T_{ab} + 12 \Theta^{a} \nabla_{a}\Xi.
\label{WEZQ5}
\end{equation}

In the case of $P_{abcd}$, application of $\nabla^h$ together with equations
\eqref{ZQWE6}, \eqref{ZQWE3}, \eqref{DeltaLambdaIdentity}, as well as the
various symmetries of $\Lambda_{abc}$ and $P^a{}_{bcd}$ results, after a rather
direct calculation in:
\begin{align}
\square P_{egmn} = & \tfrac{1}{3} R P_{egmn} - 2L_{[m}{}^{h} P_{n]heg} +
2\Lambda_{[n|eg|} \nabla_{m]}\Xi + 2\Xi \nabla_{[m}\Lambda_{n]eg} +
2\nabla_{[m}\Delta_{|eg|n]} + 2 \Xi \nabla_{[e}\Lambda_{g]mn} \nonumber \\ 
& + 2 \nabla_{[e}\Delta_{|mn|g]} - 2
\Lambda_{[e|mn|}\nabla_{g]}\Xi - 2 \Xi d_{[e}{}^{h}{}_{g]}{}^{a}P_{mnha}
- 2 \Xi d_{[e}{}^{h}{}_{|m|}{}^{a}P_{g]hna} + 2 \Xi d_{[e}{}^{h}{}_{|n|}{}^{a}P_{g]hma} \nonumber \\
& - 2 L_{[e}{}^{h}P_{g]hmn} - 2 P_{[e}{}^{h}{}_{g]}{}^{a}P_{mnha} 
- 4 P_{[e}{}^{h}{}_{|m|}{}^{a}P_{g]hna} + 2 \Xi g_{[e|m}\nabla^{h}\Lambda_{n|g]h} 
- 2 \Xi g_{[e|n}\nabla^{h}\Lambda_{m|g]h} \nonumber \\
& + 2 \Upsilon^{ha}d_{[e|hma|}g_{g]n} - 2\Upsilon^{ha}d_{[e|hna|}g_{g]m} 
+ 2\Lambda_{[g|nh|}g_{e]m} + 2 \Lambda_{n[g|h|}g_{e]m} + 2\Lambda_{m[e|h|}g_{g]n} \nonumber \\
& + 2\Lambda_{[e|mh|}g_{g]n} - 4 L^{ha}P_{[e|hma|}g_{g]n} + 4L^{ha}P_{[e|hna|}g_{g]m}.
\label{WEZQ6}
\end{align}

\subsubsection{Equation for $\Lambda_{abc}$} 

Notice that the integrability condition for $\Lambda_{abc}$, equation
\eqref{ZQWE4}, contains derivatives of zero-quantities on both sides of the
equation. This feature seems to hinder our standard approach for the
construction of a subsidiary equation. Then, in order to construct a suitable
wave equation it will be necessary to exploit the symmetries of
$\Lambda_{abcde}$.  Applying $\nabla^e$ to the integrability condition
\eqref{ZQWE4} and commuting derivatives leads to
\begin{align}
\square \Lambda_{gmn} = & \
\Lambda^{c}{}_{mn} R_{gc} + \nabla_{g}\nabla_{c}\Lambda^{c}{}_{mn} - 2
P_{g}{}^{ceh} \nabla_{h}d_{mnce} - 2 d_{mn}{}^{ce} \nabla_{h}P_{gce}{}^{h} -
\nabla^{c}\nabla^{e}\Lambda_{gce[mn]} \nonumber \\
& - 2 \Lambda^{c}{}_{[m}{}^{e}R_{|gc|n]e} -
2 d_{[m}{}^{ceh}\nabla_{|e}P_{gh|n]c} - 2
d_{g}{}^{c}{}_{[m}{}^{e}\nabla^{h}P_{n]ech} - 2 P_{[m}{}^{ceh}\nabla_{|e}d_{gh|n]c} \nonumber \\
& - 2 P_{g}{}^{c}{}_{[m}{}^{e}\nabla^{h}d_{n]ech}. \nonumber
\end{align}
Here, the double-derivative terms put at risk the hyperbolicity of the system.
For the second derivative of $\Lambda_{abc}$ one can use
\eqref{DivergenceLambda1}, while the one involving $\Lambda_{abcde}$ can be
eliminated by recalling that this field is antisymmetric under any permutation
of the first three indices --- see \eqref{Identity:LambdaFiveIndices}.  Using
this property and commuting derivatives gives
\begin{align}
\square \Lambda_{gmn} = & - \Xi \Lambda^{c}{}_{g}{}^{e} d_{mnce} + 4 \Lambda^{c}{}_{mn} L_{gc} + 2
  d_{mnce} \Delta_{g}{}^{ce} - 2 P_{g}{}^{ceh} \nabla_{h}d_{mnce} + 2
\Upsilon_{[m}{}^{c}\nabla_{|g|}T_{n]c} \nonumber \\
& - 2 \Xi \Lambda^{c}{}_{[m}{}^{e}d_{|g|n]ce}
- 4 \Xi \Lambda^{c}{}_{[m}{}^{e}d_{|ge|n]c} - 4 \Lambda^{c}{}_{g[m}L_{n]c} +
2 \Lambda_{[m}{}^{ce}P_{|gc|n]e} + 2 \Lambda^{c}{}_{g}{}^{e}P_{[m|c|n]e} \nonumber \\
& - 2 T_{[m}{}^{ce}P_{|ge|n]c} + 2 d_{g}{}^{c}{}_{[m}{}^{e}\Delta_{n]ec} - 2
d_{[m}{}^{ceh}\nabla_{|e}P_{gh|n]c} - 2 P_{[m}{}^{ceh}\nabla_{|e}d_{gh|n]c} - 2
T_{[m}{}^{c}\nabla_{|g|}\Upsilon_{n]c} \nonumber \\
& - \Xi \Lambda^{ceh}d_{[m|ceh}g_{g|n]} - 4
\Lambda^{c}{}_{[m}{}^{e}L_{|ce}g_{g|n]} -  \Lambda^{ceh}P_{[m|ceh}g_{g|n]}.
\label{WEZQ4}
\end{align}

The results of this section can be summarised in the following lemma:

\begin{lemma}
\label{Lemma:SubsidiarySystem}
Assume that the conformal fields satisfy equations \eqref{EquationMetric} and
\eqref{CWE1}-\eqref{CWE4Alternative}. Then, the geometric zero-quantities
\eqref{ZQ1}-\eqref{ZQ6} satisfy the homogeneous system of geometric wave
equations \eqref{WEZQ1}-\eqref{WEZQ4}.
\end{lemma}

\subsection{Propagation of the constraints}

As it will be discussed in detail in Section \ref{Section:GaugeConsiderations},
the system of geometric wave equations \eqref{WEZQ1}-\eqref{WEZQ4} implies, in
turn, a system of proper (hyperbolic) wave equations for which a theory of the
existence and uniqueness of solutions is readily available --- see e.g.
\cite{HugKatMar77}.  From the latter one directly obtains the following result:

\begin{proposition}
\label{Proposition:PropagationOfTheConstraints}
Assume that the geometric zero-quantities and their first derivatives vanish on a
fiduciary spacelike hypersurface $\mathcal{S}_\star$ of an unphysical spacetime
$(\mathcal{M},g_{ab})$. Then, the geometric zero-quantities vanish on the domain of
dependence $D(\mathcal{S}_\star)$ of $\mathcal{S}_\star$.
\end{proposition} 

\begin{remark}
{\em Working, for example, with coordinates adapted to the hypersurface
$\mathcal{S}_\star$, it can be readily checked that the completely spatial parts
of the zero-quantities $\Upsilon_{ab}$, $\Theta_a$, $\Lambda_{abc}$,
$\Delta_{abc}$, $Z$ and $P^a{}_{bcd}$ encode the same information as the
conformal Einstein constraint equations --- see e.g. \cite{CFEBook}, Chapter
11. Similarly, projections with a transversal (i.e. timelike) component can be
read as a first order evolution system for the geometric conformal fields ---
we ignore null components as these can be obtained as linear combinations of
transversal and intrinsic components. Thus, in order to ensure the vanishing of
the zero-quantities on the initial hypersurface $\mathcal{S}_\star$, one needs,
firstly, to produce a solution to the conformal constraint equations; this
ensures the vanishing of the spatial part of the zero-quantities. Secondly, one
reads the transversal components of the zero-quantities as definitions for the
normal derivatives of the conformal fields which can be readily computed from
the solution to the conformal constraints. In this way, the transversal
components of the zero-quantities vanish \emph{a fortiori}.}
\end{remark}

\section{Gauge considerations}
\label{Section:GaugeConsiderations}

This section provides a brief overview of the gauge freedom inherent to the
conformal Einstein field equations and the associated evolution equations. This
gauge freedom is of two types: \emph{conformal} and \emph{coordinate}. The
discussion in this section follows closely Section 2.3 in \cite{CarVal18b} and
is provided for completeness and to ease the reading of the article. 

\subsection{Conformal gauge source functions}
\label{Section:ConformalGaugeSourceFunctions}
An important feature of the conformal Einstein field equations is that the
Ricci scalar $R$ of the metric $g_{ab}$ can be regarded as a \emph{conformal
gauge source} specifying the representative in the conformal class
$[\tilde{\bmg}]$ of the (conformal) unphysical metric. Accordingly, one can
always find (locally) a conformal rescaling such that the metric $g'_{ab}$ has
a prescribed Ricci scalar $R'$. 

\begin{remark}
\label{Remark:ConformalGaugeSourceFunction}
{\em Based on the previous discussion, in what follows the Ricci
  scalar of the metric $g_{ab}$ is regarded as a prescribed function $\mathcal{R}(x)$
  of the coordinates, so one writes
\[
R=\mathcal{R}(x).
\] }
\end{remark}

\subsection{Generalised harmonic coordinates and the reduced Ricci
  operator}
\label{Section:GeneralisedWaveCoordinates}

The components of the Ricci
tensor $R_{ab}$ can be explicitly written in terms of the components
of the metric tensor $g_{ab}$ in general coordinates $x=(x^\mu)$ as 
\[
R_{\mu\nu} =
  -\frac{1}{2}g^{\lambda\rho} \partial_\lambda \partial_\rho
  g_{\mu\nu} + g_{\sigma(\mu}\nabla_{\nu)} \Gamma^\sigma + g_{\lambda\rho}
  g^{\sigma\tau} \Gamma^\lambda{}_{\sigma\mu} \Gamma^\rho{}_{\tau\nu}
  + 2 \Gamma^\sigma{}_{\lambda\rho} g^{\lambda\tau} g_{\sigma(\mu} \Gamma^\rho{}_{\nu)\tau},
\]
with
\[
\Gamma^\nu{}_{\mu\lambda} \equiv \frac{1}{2}g^{\nu\rho} ( \partial_\mu
g_{\rho\lambda} + \partial_\lambda g_{\mu\rho} - \partial_\rho g_{\mu\lambda}),
\]
where we have defined the \emph{contracted Christoffel symbols} as
$\Gamma^\nu \equiv g^{\mu\lambda} \Gamma^\nu{}_{\mu\lambda}$. A direct
computation then gives $\square x^\mu = - \Gamma^\mu$. Following the
well-known procedure for the hyperbolic reduction of the Einstein
field equations, we introduce \emph{coordinate gauge source functions}
$\mathcal{F}^\mu(x)$ to prescribe the value of the contracted
Christoffel symbols via the condition $\Gamma^\mu = \mathcal{F}^\mu(x)$.
This means that the coordinates $x=(x^\mu)$ satisfy the \emph{generalised
  wave coordinate condition}
\begin{equation}
\square x^\mu = -\mathcal{F}^\mu(x)
\label{GeneralisedWaveCoordinates}
\end{equation}
--- see e.g. \cite{Cho08,Ren08,CFEBook}.  Associated to the latter, it
is convenient to define the \emph{reduced Ricci operator}
$\mathscr{R}_{\mu\nu}[\bmg]$ as
\begin{equation}
\mathscr{R}_{\mu\nu}[\bmg] \equiv R_{\mu\nu} - g_{\sigma(\mu}\nabla_{\nu)}\Gamma^\sigma +
g_{\sigma(\mu}\nabla_{\nu)} \mathcal{F}^\sigma(x).
\label{DefinitionReducedRicci}
\end{equation}
More explicitly, one has that
\[
\mathscr{R}_{\mu\nu}[\bmg] =
  -\frac{1}{2}g^{\lambda\rho} \partial_\lambda \partial_\rho
  g_{\mu\nu} - g_{\sigma(\mu}\nabla_{\nu)} \mathcal{F}^\sigma(x) + g_{\lambda\rho}
  g^{\sigma\tau} \Gamma^\lambda{}_{\sigma\mu} \Gamma^\rho{}_{\tau\nu}
  + 2 \Gamma^\sigma{}_{\lambda\rho} g^{\lambda\tau} g_{\sigma(\mu} \Gamma^\rho{}_{\nu)\tau}.
\]
Thus, by choosing coordinates satisfying the generalised wave coordinates
condition \eqref{GeneralisedWaveCoordinates}, the unphysical Einstein equation
\eqref{EquationMetric} takes the form 
\begin{equation*}
\mathscr{R}_{\mu\nu}[\bmg] = 2 L_{\mu\nu} +\frac{1}{6}\mathcal{R}(x) g_{\mu\nu}.
\end{equation*}
Assuming that the components $L_{\mu\nu}$ are known, the latter is a quasilinear wave
equation for the components of the metric tensor. 

\subsubsection{The reduced wave operator}
\label{Section:ReducedWaveOperator}

The geometric wave operator $\square$ acting on tensorial fields contains
derivatives of the Christoffel symbols which, in turn, contain second order
derivatives of the components of the metric tensor. The presence of these
second order derivative terms is problematic as they destroy, in principle, the
hyperbolicity of the evolution equations \eqref{CWE3} and
\eqref{CWE4Alternative} since they enter in the principal part of the system.
However, as discussed in e.g. \cite{Pae15,CarVal18b}, the generalised wave
coordinate condition \eqref{GeneralisedWaveCoordinates} can be used to reduce
the geometric wave operator $\square$ to a proper second order hyperbolic
operator. 

\begin{definition}
\label{Definition:DefinitionBlackBoxOperator}
The reduced wave operator $\blacksquare$ acting on a covariant tensor field
$T_{\lambda\cdots\rho}$ is defined as
\begin{eqnarray*}
&& \blacksquare T_{\lambda \cdots \rho} \equiv \square
T_{\lambda\cdots\rho} +\bigg( (2 L_{\tau\lambda} + \frac{1}{6}\mathcal{R}(x)
   g_{\tau\lambda} - R_{\tau\lambda}) -
  g_{\sigma\tau}\nabla_\lambda (\mathcal{F}^\sigma(x)-\Gamma^\sigma) \bigg)T^\tau{}_{\cdots\rho}
 +\cdots\\
&& \hspace{3cm} \cdots + \bigg( (2 L_{\tau\rho} + \frac{1}{6}\mathcal{R}(x)
   g_{\tau\rho} - R_{\tau\rho}) -
  g_{\sigma\tau}\nabla_\rho (\mathcal{F}^\sigma(x)-\Gamma^\sigma)
  \bigg)T_{\lambda\cdots}{}^\tau,
\end{eqnarray*}
where $\square \equiv g^{\mu\nu}\nabla_\mu\nabla_\nu$. The action of
$\blacksquare$ on a scalar $\phi$ is simply given by
\[
\blacksquare \phi \equiv g^{\mu\nu}\nabla_\mu\nabla_\nu \phi. 
\]
\end{definition}

\begin{remark}
\label{Remark:DefinitionProper}
{\em The operator $\blacksquare$ provides a proper second order hyperbolic
operator for systems which involve the metric as an unknown, in contrast to
$\square$.  Accordingly, when working in generalised harmonic coordinates, all
the second order derivatives of the metric tensor can be removed from the
principal part of geometric wave equations. A system of evolutions equations
expressed in terms of the reduced wave operator $\blacksquare$ (rather than in
terms of the geometric wave operator $\square$) will be said to be
\emph{proper}.}
\end{remark}

\subsection{Summary: gauge reduced evolution equations}

The discussion of the previous sections leads us to consider the
following \emph{gauge reduced} system of evolution equations for the
components of the conformal fields $\Xi$, $s$, $L_{ab}$,
$d_{abcd}$ and $g_{ab}$ with respect to coordinates $x=(x^\mu)$
satisfying the generalised wave coordinate condition \eqref{GeneralisedWaveCoordinates}:
\begin{subequations}
\begin{eqnarray}
&&\hspace{-13mm} \blacksquare \Xi = 4s -\frac{1}{6}\Xi \mathcal{R}(x), \label{ReducedWaveCFE1} \\
&&\hspace{-13mm} \blacksquare s = - \tfrac{1}{6} s \mathcal{R}(x) + \Xi L_{\mu\nu} L^{\mu\nu} 
- \tfrac{1}{6} \nabla_{\mu}\mathcal{R}(x)
\nabla^{\mu}\Xi + \tfrac{1}{4} \Xi^5 T_{\mu\nu}T^{\mu\nu} - \Xi^3 L_{\mu\nu}T^{\mu\nu} 
+ \Xi \nabla^{\mu}\Xi \nabla^{\nu}\Xi T_{\mu\nu},
\label{ReducedWaveCFE2} \\
&&\hspace{-13mm} \blacksquare L_{\mu\nu} = -2 \Xi d_{\mu\rho\nu\lambda} L^{\rho\lambda} 
+ 4 L_{\mu}{}^{\lambda} L_{\nu\lambda} - L_{\lambda\rho} L^{\lambda\rho} g_{\mu\nu} 
+ \tfrac{1}{6} \nabla_{\mu}\nabla_{\nu}\mathcal{R}(x) + \tfrac{1}{2} \Xi^3 d_{\mu\lambda\nu\rho} T^{\lambda\rho} 
\nonumber \\
&& \hspace{0.1cm} - \Xi \nabla_{\lambda}T_{\mu}{}^{\lambda}{}_{\nu} 
- 2 T_{(\mu|\lambda|\nu)} \nabla^{\lambda}\Xi,\label{ReducedWaveCFE3} \\
&&\hspace{-13mm} \blacksquare d_{\mu\nu\lambda\rho} = - 4 \Xi
d_{\mu}{}^{\tau}{}_{[\lambda}{}^{\sigma} d_{\rho]\sigma\nu\tau} - 2 \Xi
d_{\mu}{}^{\tau}{}_{\nu}{}^{\sigma} d_{\lambda\rho\tau\sigma} + \tfrac{1}{2}
d_{\mu\nu\lambda\rho} \mathcal{R}(x) - T_{[\mu}{}^\sigma \Xi^2 d_{\nu]\sigma\lambda\rho} -
\Xi^2 T_{[\lambda}{}^\sigma d_{\rho]\sigma\mu\nu} \nonumber \\
&& \hspace{0.3cm} - \Xi^2 g_{\mu[\lambda} d_{\rho]\sigma\nu\tau} T^{\tau\sigma} 
+ \Xi^2 g_{\nu[\lambda}  d_{\rho]\sigma\mu\tau} T^{\tau\sigma} 
+ 2 \nabla_{[\mu}T_{|\lambda\rho|\nu]} + \epsilon_{\mu\nu\sigma\tau} \nabla^{\tau}\,^*T_{\lambda\rho}{}^{\sigma},
\label{ReducedWaveCFE4}\\
&& \hspace{-13mm}\mathscr{R}_{\mu\nu}[\bmg] = 2 L_{\mu\nu}
   +\frac{1}{6}\mathcal{R}(x) g_{\mu\nu}. \label{ReducedWaveCFE5}
\end{eqnarray}
\end{subequations}

\begin{remark}
\label{Remark:WaveEquationsExplicit}
{\em The reduced system of evolution equations 
  \eqref{ReducedWaveCFE1}-\eqref{ReducedWaveCFE5} is a system
of quasilinear wave equations for the fields $\Xi$, $s$,
$L_{\mu\nu}$, $d_{\mu\nu\lambda\rho}$ and $g_{\mu\nu}$. More
explicitly, one has that
\begin{eqnarray*}
&& g^{\sigma\tau}\partial_\sigma\partial_\tau \Xi = X
\big(\bmg,\partial\bmg, \Xi,s,\mathcal{R}(x)\big), \\
&& g^{\sigma\tau}\partial_\sigma\partial_\tau s =
S\big(\bmg,\partial\bmg,\Xi, \partial\Xi,s,{\bm
  L},\mathcal{R}(x),\partial\mathcal{R}(x), \bmT \big), \\
&& g^{\sigma\tau}\partial_\sigma\partial_\tau L_{\mu\nu}
=F_{\mu\nu}\big(\bmg,\partial\bmg,\Xi,{\bm L},{\bm d},
\mathcal{R}(x),\partial^2\mathcal{R}(x), \bmT, \partial\bmT\big), \\
&&  g^{\sigma\tau}\partial_\sigma\partial_\tau
d_{\mu\nu\lambda\rho}=D_{\mu\nu\lambda\rho}\big(\bmg,\partial\bmg,\Xi,{\bm
  d},\mathcal{R}(x), \partial\bmT\big), \\
&& g^{\sigma\tau}\partial_\sigma\partial_\tau g_{\mu\nu}
=G_{\mu\nu}\big(\bmg,\partial\bmg,{\bmL},\mathcal{R}(x)\big),
\end{eqnarray*}
where $X$, $S$, $F_{\mu\nu}$, $D_{\mu\nu\lambda\rho}$ and $G_{\mu\nu}$ are
polynomial expressions of their arguments. Strictly speaking, the system is a
system of wave equations only if $g_{\mu\nu}$ is known to be Lorentzian. The
basic existence, uniqueness and stability results of systems of the above type
have been given in \cite{HugKatMar77} --- these results are the second order
analogues of the theory developed in \cite{Kat75} for symmetric hyperbolic
systems. The basic theory for initial-boundary value problems can be found in
\cite{CheWah83,DafHru85}.}
\end{remark}

\section{Propagation of the gauge}
\label{Section:PropOfGauge}

This section is devoted to studying the consistency of the conformal and
coordinate gauge introduced in \Cref{Section:GaugeConsiderations} by
constructing a system of homogeneous wave equations for a set of subsidiary
fields.  The coming discussion extends the analysis in \cite{CarVal18b},
Section 5, for the vacuum case which is closely followed --- accordingly, we
mainly focus on the new features arising from the presence of matter. 

\subsection{Basic relations}

Consider a set of coordinates $x=(x^\mu)$. Let $g_{\mu\nu}$ denote the
components of a metric $g_{ab}$ in these coordinates. Similarly, $R_{\mu\nu}$
denotes the components of the associated Ricci tensor $R_{ab}$, while $R$ is
the corresponding Ricci scalar. We now investigate the requirements for $R$ and
$R_{\mu\nu}$ to coincide, respectively, with $\mathcal{R}(x)$ and
$\mathscr{R}_{\mu\nu}$.  In addition, we also need to investigate the
conditions under which $L_{\mu\nu}$ corresponds to the components of the
Schouten tensor. This can be expressed as the vanishing of the following
fields:
\begin{subequations}
\begin{eqnarray}
&& Q \equiv R -\mathcal{R}(x), \label{QScalar} \\
&& Q^\mu \equiv \Gamma^\mu - \mathcal{F}^\mu(x), \label{QVector} \\
&& Q_{\mu\nu} \equiv R_{\mu\nu} - \mathscr{R}_{\mu\nu}. 
\label{QTensor}
\end{eqnarray}
\end{subequations}
We make the following assumption:
\begin{assumption}
\label{Assumption:AssumptionGauge}
Let $T_{\mu\nu}$ and $T_{\mu\nu\lambda}$ be, respectively, the components of a
tracefree energy momentum tensor with vanishing divergence and its associated
rescaled Cotton tensor.  Let $g_{\mu\nu}$ and $L_{\mu\nu}$ be solutions to the
equations:
\begin{subequations}
\begin{eqnarray}
&& \mathscr{R}_{\mu\nu} = 2L_{\mu\nu} + \tfrac16\mathcal{R}(x)g_{\mu\nu}, \label{ReducedRicciGauge} \\
&& \blacksquare L_{\mu\nu} = -2 \Xi d_{\mu\rho\nu\lambda} L^{\rho\lambda} 
+ 4 L_{\mu}{}^{\lambda} L_{\nu\lambda} - L_{\lambda\rho} L^{\lambda\rho} g_{\mu\nu} 
+ \tfrac{1}{6} \nabla_{\mu}\nabla_{\nu}\mathcal{R}(x) + \tfrac{1}{2} \Xi^3 d_{\mu\lambda\nu\rho} T^{\lambda\rho} 
\nonumber \\
&& \hspace{1.3cm} - \Xi \nabla_{\lambda}T_{\mu}{}^{\lambda}{}_{\nu} -
   2 T_{(\mu|\lambda|\nu)} \nabla^{\lambda}\Xi. \label{BlackBoxPhiGauge}
\end{eqnarray}
\end{subequations}
\end{assumption}
As a direct consequence of equation \eqref{ReducedRicciGauge}, one can find that
the gauge zero-quantities \eqref{QScalar}-\eqref{QTensor} are not
independent of each other. Simple calculations yield
\begin{subequations}
\begin{eqnarray}
&& Q_{\mu\nu} = \nabla_{(\mu}Q_{\nu)}, \label{QIdentities1} \\
&& Q = Q_\mu{}^\mu = \nabla_\mu Q^\mu. \label{QIdentities2}
\end{eqnarray}
\end{subequations}
Furthermore, equation \eqref{DefinitionReducedRicci} and
\Cref{Definition:DefinitionBlackBoxOperator} lead to 
\begin{subequations}
\begin{eqnarray}
&& \mathscr{R}_{\mu\nu}[\bmg] = R_{\mu\nu} -\nabla_{(\mu} Q_{\nu)}, \label{ReducedRicci}\\
&& \blacksquare L_{\mu\nu} = \square L_{\mu\nu} - (Q_{\mu\sigma}
-\nabla_\mu Q_\sigma) L^\sigma{}_\nu - (Q_{\nu\sigma} -\nabla_\nu
Q_\sigma)L^\sigma{}_\mu. \label{ReducedWave}
\end{eqnarray}
\end{subequations}
\begin{remark}
\em{ Equations \eqref{QIdentities1}-\eqref{QIdentities2} show that if $Q^\mu =0$
then $Q$ and $Q_{\mu\nu}$ automatically vanish. In this sense, we will consider
$Q^\mu$ as the basic gauge zero-quantity of the system.}
\end{remark}

\subsection{The gauge subsidiary evolution system}

In this subsection we obtain a system of homogeneous wave equations for the
gauge subsidiary variables. This will be achieved by exploiting the properties
of the so-called \emph{Bach tensor} which will play the role of an
integrability condition for the system.

\subsubsection{The Bach tensor}

Let $g_{ab}$ be a 4-dimensional metric. The Bach tensor is defined as:
\begin{equation}
B_{ab} \equiv \nabla^c\nabla_a L_{bc} - \nabla^c\nabla_c L_{ab} - C_{acdb}L^{cd}.
\end{equation}
From this definition it is easy to verify that $B_{ab}$ is symmetric and
tracefree. Additionally, it satisfies the following identity, independently
of the validity of the Einstein field equations:
\begin{equation}
\nabla^aB_{ab} = 0. \label{DivergenceBach}
\end{equation}

\begin{remark}
{\em A straightforward calculation shows that the Bach tensor can be expressed
in terms of the geometric zero-quantities as
\begin{equation*}
B_{ab} = - L^{cd} P_{acbd} -  \tfrac{1}{2} \Xi^3 d_{acbd} T^{cd} + \Xi
\nabla_{c}T_{a}{}^{c}{}_{b} + 2T_{(a|c|b)} \nabla^{c}\Xi.
\end{equation*}
Consequently, if $g_{ab}$ is a solution to the tracefree metric
conformal Einstein field equations then the Bach tensor vanishes if $T_{ab} =
0$.}
\end{remark}

\begin{remark}
{\em In view of the fact that trivial initial conditions for the
zero-quantities imply the vanishing of $P^a{}_{bcd}$ --- see
\Cref{Proposition:PropagationOfTheConstraints} --- throughout the remainder of the
article, and for the sake of simplicity, our calculations will assume that
$P^a{}_{bcd} = 0$.}
\end{remark}

\subsubsection{Wave equations for the gauge subsidiary variables}

The Bach tensor can be conveniently expressed in terms of the gauge
zero-quantities.  Terms containing $R_{\mu\nu}$ and $R$ can be rewritten
according to definitions \eqref{QScalar} and \eqref{QTensor} along with
\eqref{QIdentities1} and \eqref{ReducedRicci}. A procedure similar to
that of Section 5.2 in \cite{CarVal18b} allows us to show that the Bach
tensor can be expressed in the form
\begin{equation}
B_{\mu\nu} = B'_{\mu\nu} +N_{\mu\nu},
\label{BachSplit}
\end{equation}
where $B'_{\mu\nu}$ is an expression homogeneous on $Q$, $Q_\mu$, $Q_{\mu\nu}$
and its derivatives up to fourth order and which is identical to the one found
in \cite{CarVal18b}. Here, the contributions from $T_{\mu\nu}$
have been grouped in the symmetric tensor
\[
N_{\mu\nu} \equiv - \tfrac12\Xi^3 d_{\mu\lambda\nu\rho}T^{\lambda\rho} + 2
T_{(\mu|\lambda|\nu)}\nabla^\lambda\Xi + \Xi \nabla_\lambda
T_\mu{}^\lambda{}_\nu.
\]

Next, we introduce the auxiliary field
\begin{equation}
M_\mu \equiv \square Q_\mu. \label{WEGaugeM}
\end{equation}
Taking the divergence of equation \eqref{BachSplit}, and after some direct manipulations,
equations \eqref{QIdentities1}-\eqref{QIdentities2} and
\eqref{DivergenceBach} imply that
\[
\square M_\mu = H_\mu (\nabla\bmM, \nabla\bmQ, \nabla Q, \bmQ, Q) + 4\nabla^\nu N_{\nu\mu},
\]
where $\bmQ$ stands for $Q_\mu$ and, for simplicity, $H_\nu$ represents a
homogeneous function of its arguments.  On the other hand, we can rewrite the
term $\nabla^\nu N_{\nu\mu}$ in a suitable way by using the symmetries of
$T_{abc}$ along with the help of equations \eqref{QTensor},
\eqref{ReducedRicci} and the geometric zero-quantities. A direct calculation
shows that
\[
\nabla^\nu N_{\nu\mu} = -T_{\mu\nu\lambda}\Upsilon^{\nu\lambda}
- \tfrac12 \Xi^3 T_{\nu\lambda}\Lambda^\nu{}_\mu{}^\lambda,
\]
so the wave equation for $M_\mu$ takes the schematic form
\begin{equation}
\square M_\mu = H_\mu (\nabla\bmM, \nabla\bmQ, \nabla Q, \bmQ, Q, \bmUpsilon, \bmLambda).
\label{WEGaugeQa}
\end{equation}
Lastly, a wave equation for $Q$ is required to close the system. This can be
obtained by direct application of the $\square$ operator on the definition of
$Q$ along with the aid of equations \eqref{QScalar}, \eqref{QIdentities2} and
\eqref{ReducedRicci}, resulting in
\begin{equation}
\square Q = -2L_{\mu\nu}\nabla^\mu Q^\nu - \nabla^\mu Q^\nu \nabla_{(\mu}Q_{\nu)}
- \tfrac12 Q^\mu\nabla_\mu Q - \tfrac12 Q^\mu\nabla_\mu \mathcal{R}(x) - \tfrac16\mathcal{R}(x)Q
+ \nabla^\mu M_\mu.\label{WEGaugeQ}
\end{equation}

\begin{remark}
{\em The gauge subsidiary evolution system, equations
\eqref{WEGaugeM}-\eqref{WEGaugeQ}, is homogeneous in $M_\mu, \ Q_\mu, \ Q, \
\Upsilon_{\mu\nu}$, $\Lambda_{\mu\nu\lambda}$ and their first derivatives. }
\end{remark}

The previous discussion leads to the following result:

\begin{lemma}
\label{Lemma:InitialGaugeConditions}
Assume that the hypotheses of Lemma \ref{Lemma:SubsidiarySystem} hold.
Moreover, let the quantities $M_\mu$, $Q_\mu$, $Q$,
$\Upsilon_{\mu\nu}$ and $\Lambda_{\mu\nu\lambda}$
along with their first covariant derivatives vanish on a fiduciary hypersurface
$\mathcal{S_\star}$. Then the unique solution to the system
\eqref{WEGaugeM}-\eqref{WEGaugeQ} on a small enough slab of
$\mathcal{S}_\star$ corresponds to $Q = 0$, $Q_\mu= 0$ and $M_\mu= 0$, which in
turn implies that $Q_{\mu\nu}=0$.
\end{lemma}

\begin{remark}
{\em As discussed in Section 5.2.3 of \cite{CarVal18b} the initial
  gauge conditions in Lemma \ref{Lemma:InitialGaugeConditions} can be
  rephrased in terms of conditions on the lapse and shift (and their
  derivatives) associated to the coordinate gauge source function
  $\mathcal{F}^\mu(x)$. It must be pointed out that these initial
  gauge conditions are not equivalent, in the vacuum case, to those
  considered in \cite{Pae15} which only require the vanishing of the
  gauge zero-quantities and their first derivatives on the initial
  hypersurface. In the present case, the conditions require the
  vanishing of third order derivatives via the definition of $M_\mu$.}
\end{remark}

\section{Evolution equations for the matter fields}
\label{Section:MatterFields}

Having settled the analysis of the \emph{geometric} part of the metric tracefree
conformal Einstein equations, we now proceed to investigate the
evolution of the subsidiary equations associated to a number of matter models
of interest: the conformally invariant scalar field, the Maxwell field and the
Yang-Mills field.

\subsection{The conformally invariant scalar field}
\label{Subsection:ScalarField}

It is well-known that the equation 
\[
\tilde{\nabla}^{a}\tilde{\nabla}_{a}\tilde{\phi}=0,
\]
where $\tilde\phi$ is a scalar field, is not conformally invariant. This
deficiency can be healed by the addition of a term involving the coupling with
the Ricci scalar, namely
\begin{equation}
\tilde{\nabla}^{a}\tilde{\nabla}_{a}\tilde{\phi}-\tfrac{1}{6}\tilde
R\tilde \phi=0.
\label{WaveEqnScalarFieldPhysical}
\end{equation}
Defining the unphysical scalar field
\[
\phi\equiv \Xi^{-1}\tilde\phi,
\]
a direct computation shows that equation \eqref{WaveEqnScalarFieldPhysical} implies
\begin{equation}
\square\phi-\tfrac{1}{6}R\phi=0.
\label{ConformallyInvariantWaveEquation}
\end{equation}
In what follows, for convenience, equation \eqref{WaveEqnScalarFieldPhysical}
will known as the \emph{conformally invariant wave equation} --- or the
\emph{conformally coupled wave equation}. The energy-momentum tensor associated
to this field takes the form
\begin{equation}
T_{ab}=\nabla_{a}\phi\nabla_{b}\phi-\tfrac{1}{4}g_{ab}\nabla_{c}\phi\nabla^{c}\phi-
\tfrac{1}{2}\nabla_{a}\nabla_{b}\phi+\tfrac{1}{2}\phi^{2}L_{ab},
\label{EMTensorScalarField}
\end{equation}
so that $\nabla^a T_{ab}=0$ holds if equation
\eqref{ConformallyInvariantWaveEquation} is satisfied. It can be readily
verified that $T_{ab}$, as given by the expression above, is also tracefree.

\begin{remark}
{\em The second derivatives of $\phi$ in equation \eqref{EMTensorScalarField}
will lead to the appearance of second and third order derivatives of the matter field in
the expression of the rescaled Cotton tensor --- see equation
\eqref{RescaledCottonTensor} --- which may affect the hyperbolicity of the
system \eqref{ReducedWaveCFE1}-\eqref{ReducedWaveCFE5}. Moreover, $T_{ab}$ is
also coupled to the geometric sector via the Schouten tensor. These
difficulties will be addressed in the sequel.}
\end{remark}

\begin{remark}
{\em The conformally invariant scalar field is related to the standard
  scalar field satisfying the wave equation through a transformation
  originally due to Bekenstein \cite{Bek74}. Thus, in principle, the
theory for the conformally invariant scalar field developed in this section
can be rephrased in terms of the standard scalar field.}
\end{remark}

\subsubsection{Auxiliary fields and the evolution equations}

We start the analysis by observing that the third order derivative terms in the
expression of the rescaled Cotton tensor for the conformally invariant scalar
field are of the form $\nabla_{[a} \nabla_{b]} \nabla_c \phi$. Using the
commutator of covariant derivatives, these terms can be transformed into first
order derivative terms according to the formula
\[
\nabla_{[a} \nabla_{b]} \nabla_c \phi = - \tfrac12 R_{abc}{}^d \nabla_d\phi.
\]
Thus, one is left with an expression for the Cotton tensor containing, at most,
second order derivatives. In order to eliminate these derivatives which,
potentially, could destroy the hyperbolic nature of the wave equations
\eqref{ReducedWaveCFE1}-\eqref{ReducedWaveCFE4}, one needs to promote the first
and second derivatives of $\phi$ as further (independent) unknowns.
Accordingly, we define
\begin{equation}
\phi_{a}\equiv\nabla_a\phi,\qquad \phi_{ab}\equiv\nabla_{a}\nabla_{b}\phi.
\label{ScalarFieldSubstitutions}
\end{equation}
Following the previous discussion, and exploiting equation \eqref{TraceCFE3},
one can write the rescaled Cotton tensor for the conformally invariant scalar
field as
\begin{eqnarray}
&& T_{abc} = \bigg(1 - \tfrac14\Xi^2\phi^2 \bigg)^{-1} \bigg(\tfrac32\Xi\phi L_{c[b}\phi_{a]}
+ \tfrac32\Xi \phi_{[b}\phi_{a]c} - \tfrac14\Xi\phi^2 d_{abcd}\nabla^d\Xi
-\tfrac14\Xi^2\phi d_{abcd}\phi^d \nonumber \\
&& \hspace{1.2cm} + \tfrac12\Xi\phi g_{c[b}L_{a]d}
+ \tfrac12 \Xi g_{c[a}\phi_{b]d}\phi^d + g_{c[b}T_{a]d}\nabla^d\Xi
+3 T_{c[b}\nabla_{a]}\Xi \bigg).
\label{RescaledCottonScalarFieldGood}
\end{eqnarray}

\medskip

We now proceed to construct suitable evolution equations for $\phi_a$ and $\phi_{ab}$
by means of a set of integrability conditions for these fields.
Firstly, the identity $\nabla_a\phi_{b} = \nabla_b\phi_{a}$ represents an
integrability condition for $\phi_a$. A wave equation then readily follows
after applying $\nabla^b$ and using equation
\eqref{ConformallyInvariantWaveEquation}:
\begin{equation}
\square\phi_a
= 2\phi^bL_{ab} + \tfrac13 R\phi_a + \tfrac16\phi\nabla_a R.
\label{WaveEqnDPhi}
\end{equation}
On the other hand, an integrability condition for $\phi_{ab}$ can be obtained
directly from its definition:
\[
2\nabla_{[c}\phi_{a]b} = \phi_d R_{cab}{}^d =
-\Xi\phi^d d_{acbd} -2\phi_{[c}L_{a]b} + 2\phi^d g_{b[c}L_{a]d}.
\]
Applying $\nabla^c$ to this relation and using equations
\eqref{TraceCFE3}, \eqref{TraceCFE4}, \eqref{ConformallyInvariantWaveEquation}
and \eqref{WaveEqnDPhi}, a straightforward calculation leads to:
\begin{align}
\square\phi_{ab} = & \tfrac{1}{2} \phi_{ab} R -  \tfrac{1}{3} R \phi L_{ab} - 2 \phi^{cd} L_{cd}
g_{ab} -  \tfrac{1}{6} \phi^{c} g_{ab} \nabla_{c}R + \tfrac{1}{6} \phi
\nabla_{(a}\nabla_{b)}R  - 2 \Xi \phi^{cd}d_{(a|c|b)d} \nonumber \\
& + 8 \phi_{(a}{}^{c}L_{b)c} + 2 \Xi \phi^{c}T_{(a|c|b)} + \tfrac{2}{3}
\phi_{(a}\nabla_{b)}R + 2 \phi^{c}\nabla_{(a}L_{b)c}
- 2 \phi^{c}d_{(acb)}{}^{d}\nabla_{d}\Xi.
\label{WaveEqnDDPhi}
\end{align}

\begin{remark}
{\em In equation \eqref{WaveEqnDDPhi} it is understood that the rescaled Cotton
tensor $T_{bca}$ is expressed in terms of the auxiliary fields $\phi_a$ and
$\phi_{ab}$ according to \eqref{RescaledCottonScalarFieldGood} so does not
contain second or higher derivatives of the fields.}
\end{remark}

\begin{remark}
\label{Remark:BoxOperatorScalarField}
{\em When coupling the wave equations \eqref{ConformallyInvariantWaveEquation},
\eqref{WaveEqnDPhi} and \eqref{WaveEqnDDPhi} to the system
\eqref{ReducedWaveCFE1}-\eqref{ReducedWaveCFE5} satisfied by the geometric
conformal fields, it is understood that the geometric wave operator $\square$
is replaced by its reduced counterpart $\blacksquare$ as discussed in
Section \ref{Section:ReducedWaveOperator}. }
\end{remark}

\subsubsection{Subsidiary equations}
To verify the consistency of our approach in dealing with the higher order
derivative terms in the rescaled Cotton tensor for the conformally invariant
scalar field we introduce the following subsidiary fields:
\begin{subequations}
\begin{eqnarray}
&& Q_{a}\equiv \phi_{a}-\nabla_{a}\phi, \label{DefinitionQVector}\\
&& Q_{ab} \equiv \phi_{ab}-\nabla_{a}\nabla_{b}\phi. \label{DefinitionQTensor}
\end{eqnarray}
\end{subequations}

\medskip
A wave equation for $Q_a$ can be obtained in a straightforward way: applying
$\square$ to definition \eqref{DefinitionQVector} and with the help of
relations \eqref{ConformallyInvariantWaveEquation} and \eqref{WaveEqnDPhi},
a short calculation yields
\begin{equation}
\square Q_{a} = \square \phi_a - \nabla_a\square\phi - R_{ab}\nabla^b\phi
=\tfrac{1}{3}RQ_{a}+2L_a{}^b Q_b \label{WaveEquationQVector}.
\end{equation}
Similarly, applying $\square$ to equation \eqref{DefinitionQTensor}, commuting
covariant derivatives and using the definitions of the geometric
zero-quantities one obtains
\begin{eqnarray}
&& \square Q_{ab}= \tfrac{1}{2} Q_{ab} R - 2 Q^{cd} L_{cd} g_{ab} -
\tfrac{1}{6} Q^{c} g_{ab} \nabla_{c}R + 2 Q^{c} \nabla_{c}L_{ab} 
- 2 \Xi Q^{cd}d_{acbd} + 8 Q_{(a}{}^{c}L_{b)c} \nonumber \\ 
&& \hspace{1.4cm}  - 2 \phi^{c}\Delta_{(a|c|b)} 
+ 4 \Xi Q^{c}T_{(a|c|b)} + 4Q^{c}\Delta_{(a|c|b)} + \tfrac{2}{3} Q_{(a}\nabla_{b)}R
- 4 Q^{c}d_{(a|c|b)}{}^{d}\nabla_{d}\Xi.
\label{WaveEquationQTensor}
\end{eqnarray}

\begin{remark}
{\em The system of wave equations \eqref{WaveEquationQVector} and
\eqref{WaveEquationQTensor} is homogeneous in $Q_a, \ Q_{ab}$ and
$\Delta_{abc}$ Thus, it follows from general uniqueness results for solutions
to wave equations that if these quantities and their derivatives vanish on an
initial hypersurface $\mathcal{S}_\star$, then necessarily $Q_a=0$ and
$Q_{ab}=0$ at least on a small enough slab around $\mathcal{S}_\star$. }
\end{remark}

\subsubsection{Summary}
The analysis of the conformally invariant scalar field can be summarised
in the following manner:

\begin{proposition}
The system of equations \eqref{ReducedWaveCFE1}-\eqref{ReducedWaveCFE5} with
rescaled Cotton tensor given by \eqref{RescaledCottonScalarFieldGood},
together with the conformally invariant wave equation
\eqref{ConformallyInvariantWaveEquation} and the auxiliary system
\eqref{WaveEqnDPhi}-\eqref{WaveEqnDDPhi} written in terms of the reduced wave
operator $\blacksquare$, constitute a \emph{proper} system of quasilinear wave
equations --- see Remark \ref{Remark:DefinitionProper}.
\end{proposition}

\subsection{The Maxwell field}
\label{Subsection:MaxwellField}

The next example under consideration is the electromagnetic field. The physical
Maxwell equations expressed in terms of the antisymmetric Faraday tensor
$\tilde{F}_{ab}$ are given by
\begin{eqnarray*}
&& \tilde{\nabla}^{a}\tilde{F}_{ab}=0,\\
&& \tilde{\nabla}_{[a}\tilde{F}_{bc]}=0.
\end{eqnarray*}
It is well-known that the Maxwell equations are conformally invariant by
defining the \emph{unphysical Faraday tensor} $F_{ab}$ as
\[
F_{ab} \equiv \tilde{F}_{ab}.
\]
From here it follows that the physical Maxwell equations imply
\begin{subequations}
\begin{eqnarray}
&& \nabla^{a}F_{ab}=0, \label{UnphysicalMaxwell1} \\
&& \nabla_{[a}F_{bc]}=0, \label{UnphysicalMaxwell2}
\end{eqnarray}
\end{subequations}
with the associated \emph{unphysical Maxwell energy-momentum tensor} given by
\begin{equation}
T_{ab}=F_{ac}F_{b}{}^{c}-\frac{1}{4}g_{ab}F_{cd}F^{cd}.
\label{UmphysicalEMTensorMaxwell}
\end{equation}
Alternatively, defining the Hodge dual $F^*_{ab}$ of the Faraday tensor as
\begin{equation}
F^*_{ab} \equiv -\frac12 \epsilon_{ab}{}^{cd}F_{cd},
\end{equation}
the second unphysical Maxwell equation \eqref{UnphysicalMaxwell2} can be
written as:
\begin{equation}
\nabla^a F^*_{ab}=0 \label{UnphysicalMaxwell2Alt}.
\end{equation}

\subsubsection{Auxiliary field and the evolution equations}

A geometric wave equation for the Faraday tensor can be obtained from
differentiation of the Maxwell equation \eqref{UnphysicalMaxwell2}, which
represents a natural integrability condition for this field. Commuting
covariant derivatives and applying equation \eqref{UnphysicalMaxwell1}, a
calculation yields
\begin{equation}
\square F_{bc}=\tfrac{1}{3}F_{bc}R-2\Xi F^{ad}d_{bacd}.
\label{WaveEquationFaraday}
\end{equation}

\medskip
From equation \eqref{RescaledCottonTensor} it follows that the rescaled
Cotton tensor contains first derivatives of $F_{ab}$. This puts at risk the
hyperbolicity of the system \eqref{ReducedWaveCFE1}-\eqref{ReducedWaveCFE4}.
In order to deal with this problem we introduce the auxiliary variable
\begin{equation}
F_{abc} \equiv \nabla_a F_{bc},
\label{DefinitionFabc}
\end{equation}
satisfying $F_{abc} = F_{a[bc]}$.  By virtue of equation
\eqref{UnphysicalMaxwell2} it also follows that $F_{[abc]}=0$.  In terms of
this quantity, it can be readily checked that the rescaled Cotton tensor for the
Maxwell field takes the form
\begin{align}
T_{abc} = & \ \Xi F_{[b}{}^d F_{a]cd} - \tfrac12 \Xi F_c{}^d F_{dab} 
+ \tfrac12 \Xi g_{c[a}F^{de}F_{b]de} - 3 F_{cd} F_{[a}{}^d \nabla_{b]}\Xi
+ F_{de}F^{de}g_{c[a}\nabla_{b]}\Xi \nonumber \\
& - g_{c[a}F_{b]}{}^e F_{de}\nabla^d\Xi.
\label{RescaledCottonMaxwell}
\end{align}
From definition \eqref{DefinitionFabc} it follows that $F_{abc}$ possesses two
independent divergences: $\nabla^a F_{abc}$ is simply the right-hand side of
wave equation \eqref{WaveEquationFaraday} whilst the other is given by
\begin{equation}
\nabla_c F_{ab}{}^c = \Xi F^{cd}d_{acbd} - \frac16 RF_{ab} + 2F_{[a}{}^c L_{b]c}, \label{DivergenceFabc2}
\end{equation}
as a direct calculation confirms.  In order to obtain an integrability
condition for $F_{abc}$, consider the expression $3\nabla_{[d}F_{|a|bc]}$.
Commuting covariant derivatives and using the first Bianchi identity for the
Weyl tensor, a straightforward calculation results in:
\begin{equation}
3\nabla_{[d}F_{|a|bc]} = -3 \Xi F_{[d}{}^e d_{|ae|bc]} + 6 F_{[db}L_{c]a}
+ 6 g_{a[d}F_b{}^e L_{c]e} \label{ICFabc}.
\end{equation}
A geometric wave equation can be obtained by applying $\nabla^d$ to the last expression and
commuting derivatives. Using equations \eqref{TraceCFE3}, \eqref{TraceCFE4},
\eqref{DefinitionCottonDual}, \eqref{TraceCFE4Alt},
\eqref{DivergenceFabc2} as well as the symmetries of $d_a{}_{bcd}$ and
$T_{abc}$ to simplify it, a long but direct calculation yields
\begin{align}
\square F_{abc} = & -2\Xi F_{a}{}^{d} T_{bcd} + 4\Xi F_{[b}{}^{d} T_{|ad|c]} 
- 2 \Xi F_{a}{}^{de} d_{bdce} - 4 \Xi F^{d}{}_{[b}{}^{e} d_{c]ead} 
+ \tfrac{1}{2} F_{abc} R + 4 F^{d}{}_{bc} L_{ad} \nonumber \\
&  - 4 F^{d}{}_{a[b} L_{c]d}
- 4 F^{d}{}_{[b}{}^{e} g_{c]a} L_{de} + \tfrac{1}{3} F_{bc} \nabla_{a}R 
- 2 F^{de} d_{ade[b}\nabla_{c]}\Xi - 4 \Xi F^{de} \nabla_{[b}d_{c]ead} \nonumber \\
&  - \tfrac{1}{3} F_{a[b} \nabla_{c]}R 
- 2 F_{[b}{}^{e} d_{c]ead} \nabla^{d}\Xi -  F_{d}{}^{e} d_{aebc} \nabla^{d}\Xi 
- 4 F_{[b}{}^{e} d_{c]dae} \nabla^{d}\Xi - F_{a}{}^{e} d_{bcde} \nabla^{d}\Xi \nonumber \\
& + 2 F^{ef} g_{a[b} d_{c]edf}\nabla^{d}\Xi + \tfrac{1}{3} g_{a[b} F_{c]d} \nabla^{d}R.
\end{align}
This equation can be further simplified via a pair of observations. Firstly,
by multiplying equation \eqref{PaetzIdentityInter} by $F^{dg}$ the following
auxiliary identity is found:
\begin{equation}
2 F_{[a}{}^{e} d_{b]ecd} \nabla^{d}\Xi - 2 F_{[c}{}^{e} d_{d]eab} \nabla^{d}\Xi 
+ 2 F^{de} d_{ced[a} \nabla_{b]}\Xi - 2 F^{eg} g_{c[a} d_{b]edg} \nabla^{d}\Xi =0.
\label{Aux1}
\end{equation}
Secondly, from equation \eqref{TraceCFE4Alt} we have the following relations:
\begin{eqnarray}
&& 4 \Xi F^{de} \nabla_{[b}d_{c]ead} = -2\Xi\epsilon_{bcef}F^{de} \,^*T_{ad}{}^f
+ 2\Xi F^{de}\nabla_e d_{adbc}, \nonumber \\
&& \Xi F^{de} \nabla_{e}d_{adbc} = 
 - \tfrac{1}{2} \Xi \epsilon_{adef} F^{de} \,^*T_{bc}{}^{f} -  \tfrac{1}{2}
  \Xi F^{de} \nabla_{a}d_{bcde}. \nonumber
\end{eqnarray}
Combining them we readily obtain the identity
\begin{equation}
4 \Xi F^{de} \nabla_{[b}d_{c]ead} =  4 \Xi F_{[b}{}^{d} T_{|a|c]d} - 2 \Xi
F_{a}{}^{d} T_{bcd} + \Xi F^{de} \nabla_{a}d_{bcde}. \label{Aux2}
\end{equation}
Making use of \eqref{Aux1} and \eqref{Aux2}, the wave equation for $F_{abc}$
takes a simpler form:
\begin{align}
\square F_{abc} = & \ 4\Xi F_{[b}{}^d T_{c]da} 
- 2 \Xi F_{a}{}^{de} d_{bdce} - 4 F^d{}_{[b}{}^e d_{c]ead}
+ \tfrac{1}{2} F_{abc} R + 4 F^{d}{}_{bc} L_{ad} - 4 F^{d}{}_{a[b} L_{c]d} \nonumber \\
& - 4 F^{d}{}_{[b}{}^{e} g_{c]a} L_{de} + \tfrac{1}{3} F_{bc} \nabla_{a}R 
- \tfrac{1}{3} F_{a[b} \nabla_{c]}R + \tfrac{1}{3} g_{a[b} F_{c]d} \nabla^{d}R
- 4 F^{de} d_{ade[b} \nabla_{c]}\Xi \nonumber \\
& - 4 F_{[b}{}^{e} d_{c]dae} \nabla^{d}\Xi
-2 F_a{}^e d_{bcde}\nabla^d\Xi - \Xi F^{de}\nabla_a d_{bcde}. \label{WaveEquationDerFaraday}
\end{align}
As remarked in the case of the conformally invariant scalar field, the
geometric operator $\square$ is to be replaced by $\blacksquare$ when equations
\eqref{WaveEquationFaraday} and \eqref{WaveEquationDerFaraday} are coupled to
the system \eqref{ReducedWaveCFE1}-\eqref{ReducedWaveCFE5}.

\subsubsection{Subsidiary equations}
In order to complete the discussion of the Maxwell field it is necessary to
construct suitable evolution equations for the zero-quantities
\begin{subequations}
\begin{eqnarray}
&& M_{b} \equiv \nabla^{a}F_{ab}, \label{DefinitionZQMaxwell1}\\
&& M_{abc}\equiv \nabla_{[a}F_{bc]}, \label{DefinitionZQMaxwell2}\\
&& Q_{abc} \equiv F_{abc} - \nabla_a F_{bc}. \label{DefinitionZQMaxwellAuxiliary}
\end{eqnarray}
\end{subequations}
Here, $M_{abc}$ possesses the symmetries
\begin{subequations}
\begin{eqnarray}
M_{abc} = M_{a[bc]} = M_{[ab]c} = M_{[abc]}.
\end{eqnarray}
\end{subequations}
Also, one can verify the following identities:
\begin{subequations}
\begin{eqnarray}
&& \nabla^a M_a = 0, \label{DivergenceMa}\\
&& \nabla^c M_{abc} = -\tfrac23 \nabla_{[a}M_{b]}. \label{DivergenceMabc}
\end{eqnarray}
\end{subequations}

\begin{remark}
{\em Following the spirit of the discussion in the previous section, the
zero-quantities $M_a$ and $M_{abc}$ encode Maxwell equations
\eqref{UnphysicalMaxwell1} and \eqref{UnphysicalMaxwell2}, respectively, while
$Q_{abc}$ does so for the auxiliary field $F_{abc}$.}
\end{remark}

\medskip

\noindent
\textbf{Equation for $M_a$.} Observe that equation \eqref{DivergenceMabc} works
as an integrability condition for $M_a$. Applying $\nabla^b$, using
\eqref{DivergenceMa} and exploiting the various symmetries of $M_{abc}$, one
obtains 
\begin{equation}
\square M_{a} = \tfrac{1}{6} M_{a}R + 2 M^{b} L_{ab}.
\label{WEZQMaxwell1}
\end{equation}

\medskip
\noindent
\textbf{Equation for $M_{abc}$.} In order to avoid lengthy expressions it is
simpler to consider the Hodge dual of $M_{abc}$ defined as
\begin{equation}
M^*_a \equiv \nabla^b F^*_{ba} = \frac12 \epsilon_a{}^{bcd} M_{bcd}.
\end{equation}
Here, the second equality is a consequence of equations \eqref{DefinitionFabc}
and \eqref{DefinitionZQMaxwell2}. From this definition it can be easily checked
that $M^*_a$ is divergencefree which, in turn, implies an integrability
condition. More explicitly:
\begin{equation}
\nabla^a M^*_a = 0 \iff \nabla_{[d}M_{abc]} = 0.
\label{ICMabc}
\end{equation}
Applying $\nabla^d$ to \eqref{ICMabc} and commuting derivatives, a straightforward
calculation leads to
\begin{eqnarray}
&& \square M_{abc} = \tfrac{1}{2} R M_{abc} - 6 \Xi
d_{[a}{}^{d}{}_{b}{}^{e}M_{c]de} - 6 L_{[a}{}^{d}M_{bc]d},
\label{WEZQMaxwell2}
\end{eqnarray}
where it has been used that $\nabla_{[a}\nabla_{|d|}M_{bc]}{}^d$ vanishes by
virtue of equation \eqref{DivergenceMabc}.

\medskip
\noindent
\textbf{Equation for $Q_{abc}$.} A wave equation for the field $Q_{abc}$ can
be obtained by direct application of the $\square$ operator. Employing
definitions \eqref{DefinitionZQMaxwell1} and \eqref{DefinitionZQMaxwellAuxiliary},
along with equations \eqref{ZQ3}, \eqref{ZQ4},
\eqref{WaveEquationFaraday} and \eqref{WaveEquationDerFaraday}, one obtains the
expression
\begin{align}
\square Q_{abc} = & \ 4 \Xi F_{[b}{}^{d} \Lambda_{|a|c]d} - 2 \Xi Q_{a}{}^{de} d_{bdce} 
- 2 \Xi Q^{d}{}_{[b}{}^{e} d_{c]ead} + \tfrac{1}{2} Q_{abc} R 
- 4 M_{[b} L_{c]a} + 4 Q^{d}{}_{bc} L_{ad} \nonumber \\
& - 4 Q^{d}{}_{a[b} L_{c]d} 
+ 6 L_{a}{}^{d} M_{bcd} - 4 Q^{d}{}_{[b}{}^{e} g_{c]a} L_{de}
+ 2 F^{de} d_{bdce} \nabla_{a}\Xi - 4 F^{de} d_{ade[b} \nabla_{c]}\Xi \nonumber \\
& - 6 F_{[a}{}^{e} d_{bc]de} \nabla^{d}\Xi.
\end{align}
In order to show that the terms not containing zero-quantities vanish,
observe that the first Bianchi identity implies that
\[
2 F^{de} d_{bdce} \nabla_{a}\Xi - 4 F^{de} d_{ade[b}\nabla_{c]}\Xi = 3 F^{de} d_{de[ab}\nabla_{c]}\Xi.
\]
On the other hand, multiplying definition $\eqref{Lambda_abcde}$ by
$F^{de}$, a short calculation yields the auxiliary identity
\begin{equation*}
3 F^{de} d_{de[ab}\nabla_{c]}\Xi - 6 F_{[a}{}^e d_{bc]de}\nabla^d\Xi = 0.
\end{equation*}
From the last two expressions it follows then that
\begin{align}
\square Q_{abc} = & \ 4 \Xi F_{[b}{}^{d} \Lambda_{|a|c]d} - 2 \Xi Q_{a}{}^{de} d_{bdce} 
- 2 \Xi Q^{d}{}_{[b}{}^{e} d_{c]ead} + \tfrac{1}{2} Q_{abc} R 
- 4 M_{[b} L_{c]a} + 4 Q^{d}{}_{bc} L_{ad} \nonumber \\
& - 4 Q^{d}{}_{a[b} L_{c]d} 
+ 6 L_{a}{}^{d} M_{bcd} - 4 Q^{d}{}_{[b}{}^{e} g_{c]a} L_{de}.
\label{WEZQMaxwellAuxiliary}
\end{align}

\begin{remark}
\em{ Geometric wave equations \eqref{WEZQMaxwell1}, \eqref{WEZQMaxwell2} and
\eqref{WEZQMaxwellAuxiliary} are crucially homogeneous in $M_a$, $M_{abc}$,
$Q_{abc}$ and $\Lambda_{abc}$.  Thus, if these quantities and their first
covariant derivatives vanish on an initial hypersurface $\mathcal{S}_\star$, it
can be guaranteed that there exists a unique solution on a small enough slab
of $\mathcal{S}_\star$, and it corresponds to $M_a = 0$, $M_{abc} = 0$ and
$Q_{abc} = 0$.}
\end{remark}

\subsubsection{Summary}

The previous discussion about the coupling of the Maxwell field to the
metric tracefree conformal Einstein field equations can be summarised as
follows:

\begin{proposition}
The system of wave equations \eqref{ReducedWaveCFE1}-\eqref{ReducedWaveCFE5}
with rescaled Cotton tensor given by \eqref{RescaledCottonMaxwell} together
with the wave equations \eqref{WaveEquationFaraday} and
\eqref{WaveEquationDerFaraday} written in terms of the wave operator
$\blacksquare$ is a proper quasilinear system of wave equations for the
Einstein-Maxwell system.
\end{proposition}

\subsection{The Yang-Mills field}
\label{Subsection:YangMillsField}

The Yang-Mills field is the last example of a tracefree matter model we study
in this paper. Due to its similarities with the Faraday field, some of the
calculations will result analogous to the ones performed in the previous
subsection. However, one of the distinctive features of the Yang-Mills field is
the fact that, in order to obtain a hyperbolic reduction of the equations, one
needs to introduce a set of gauge source functions fixing the divergence of the
gauge potential.  The consistency of this gauge choice will be analysed towards
the end of the section.

\subsubsection{Basic equations}

The \emph{Yang-Mills field} consists of a set of \emph{gauge potentials}
$\tilde{A}^\fraka{}_\mu$ and \emph{gauge fields} $\tilde{F}^\fraka{}_{\mu\nu}$
where the indices $\fraka,\;\frakb,\dots$ take values in a Lie algebra
$\mathfrak{g}$ of a Lie group $\mathfrak{G}$. The equations satisfied by the
fields $\tilde{A}^\fraka{}_\mu$ and $\tilde{F}^\fraka{}_{\mu\nu}$ are:
\begin{eqnarray*}
&& \tilde{\nabla}_a \tilde{A}^\fraka{}_b - \tilde{\nabla}_a \tilde{A}^\fraka{}_b +
C^\fraka{}_{\frakb\frakc} \tilde{A}^\frakb{}_a \tilde{A}^\frakc{}_b -\tilde{F}^\fraka{}_{ab} =0, \\
&& \tilde{\nabla}^a \tilde{F}^\fraka{}_{ab} + C^\fraka{}_{\frakb\frakc}
\tilde{A}^{\frakb a} \tilde{F}^\frakc{}_{ab}=0, \\
&& \tilde{\nabla}_{[a} \tilde{F}{}^\fraka{}_{bc]} + C^\fraka{}_{\frakb\frakc}
\tilde{A}^\frakb{}_{[a} \tilde{F}^\frakc{}_{bc]} =0.
\end{eqnarray*}
Here $C^\fraka{}_{\frakb\frakc}=C^\fraka{}_{[\frakb\frakc]}$ denote the
structure constants of the Lie algebra $\mathfrak{g}$ which satisfy the
\emph{Jacobi identity}
\begin{equation}
C^\fraka{}_{\frakd\frake} C^\frakb{}_{\fraka\frakc}
+C^\fraka{}_{\frake\frakc}C^\frakb{}_{\fraka\frakd}+C^\fraka{}_{\frakc\frakd}C^\frakb{}_{\fraka\frake}=0.
\label{JacobiIdentityStructureConstants}
\end{equation}
Also, the energy-momentum tensor associated to the Yang-Mills field is given by
\[
\tilde{T}_{ab} = \tfrac{1}{4} \delta_{\fraka\frakb}\tilde{F}^\fraka{}_{cd}
\tilde{F}^{\frakb cd} \tilde{g}_{ab}
-\delta_{\fraka\frakb}\tilde{F}^\fraka{}_{ac}\tilde{F}^\frakb{}_b{}^c.
\]

\medskip
\noindent
\textbf{Conformal invariance.} The Yang-Mills equations are well-known
to be conformally invariant. More precisely, defining the
\emph{unphysical fields}:
\[
F^\fraka{}_{ab} \equiv \tilde{F}^\fraka{}_{ab}, \qquad A^\fraka{}_{a}
\equiv \tilde{A}^\fraka{}_{a},
\]
a direct computation under the conformal rescaling \eqref{Definition:ConformalRescaling} renders
the \emph{unphysical Yang-Mills equations}
\begin{subequations}
\begin{eqnarray}
&& \nabla_a A^\fraka{}_b - \nabla_b A^\fraka{}_a +
C^\fraka{}_{\frakb\frakc} A^\frakb{}_a A^\frakc{}_b -F^\fraka{}_{ab} =0, \label{UnphysicalYM1}\\
&& \nabla^a F^\fraka{}_{ab} + C^\fraka{}_{\frakb\frakc}
A^{\frakb a} F^\frakc{}_{ab}=0, \label{UnphysicalYM2}\\
&& \nabla_{[a} F{}^\fraka{}_{bc]} + C^\fraka{}_{\frakb\frakc}
A^\frakb{}_{[a} F^\frakc{}_{bc]} =0. \label{UnphysicalYM3}
\end{eqnarray}
\end{subequations}
In addition, the unphysical energy-momentum tensor is
\begin{equation}
T_{ab} = \tfrac{1}{4} \delta_{\fraka\frakb}F^\fraka{}_{cd}
F^{\frakb cd} g_{ab}
-\delta_{\fraka\frakb}F^\fraka{}_{ac}F^\frakb{}_b{}^c. \label{EMTensorYM}
\end{equation}
Finally, it will result useful to introduce the dual of $F^\fraka{}_{ab}$
defined as
\begin{equation}
F^{*\fraka}{}_{ab} \equiv -\tfrac12\epsilon_{ab}{}^{cd}F^\fraka{}_{cd}.
\label{FaradayYMDual}
\end{equation}

\begin{remark}
\em{Due to the form of the energy-momentum tensor given in \eqref{EMTensorYM},
first derivatives of $F^\fraka{}_{ab}$  will appear in the rescaled Cotton tensor,
putting at risk the hyperbolicity of the system
\eqref{ReducedWaveCFE1}-\eqref{ReducedWaveCFE4}.  As in the case of the
Maxwell field, this will make necessary the introduction of an auxiliary quantity.}
\end{remark}

\subsubsection{Evolution equations for the Yang-Mills fields} Suitable wave
equations for the Yang-Mills fields can be obtained by a procedure analogous to
the one used for the Maxwell field. Accordingly, we introduce the auxiliary field
\begin{equation}
F^\fraka{}_{abc} \equiv \nabla_a F^\fraka{}_{bc} + C^\fraka{}_{\frakb\frakc}A^\frakb{}_a
F^\frakc{}_{bc}.
\label{DefinitionFabcYM}
\end{equation}
Moreover, the construction of a geometric wave equation for the Yang-Mills
gauge potentials requires the introduction of \emph{gauge source functions}
$f^\fraka(x)$ depending in a smooth way on the coordinates and fixing the value
of the divergence of the potential. More precisely, in the following we set
\begin{equation} 
\nabla^a A^\fraka{}_a \equiv f^\fraka(x). \label{Definition:YMGaugeSourceFunction}
\end{equation}

\medskip
\noindent
\textbf{Equation for the field strength.} The Yang-Mills Bianchi identity,
equation \eqref{UnphysicalYM3},  represents an integrability condition for the
field strength tensors $F^\fraka{}_{ab}$. Differentiating it and making use of
equations \eqref{JacobiIdentityStructureConstants} and
\eqref{UnphysicalYM1}-\eqref{UnphysicalYM3}, a straightforward calculation
results in
\begin{align}
\square F^{\fraka}{}_{ab} = & -2 \Xi F^{\fraka cd}
d_{acbd} + \tfrac{1}{3} F^{\fraka}{}_{ab} R +
2C^\fraka{}_{\frakb\frakc} F^\frakb{}_a{}^c F^\frakc{}_{bc} 
- 2C^\fraka{}_{\frakb\frakc} F^\frakc{}_{cab}A^{\frakb c}
- C^\fraka{}_{\frakb\frake} C^\frake{}_{\frakc\frakd}
 F^\frakd{}_{ab}A^{\frakb c} A^\frakc{}_c \nonumber \\
&  + C^\fraka{}_{\frakb\frakc} f^\frakb(x) F^\frakc{}_{ab}.
\label{CWEFaradayYM}
\end{align}

\medskip
\noindent
\textbf{Equation for the gauge potential.} Equation \eqref{UnphysicalYM1}
provides a natural integrability condition for the gauge potential field. After
applying $\nabla^b$, commuting derivatives and using equation
\eqref{UnphysicalYM2}, one arrives to:
\begin{equation}
\square A^{\fraka}{}_{a} = \tfrac{1}{6} A^{\fraka}{}_{a} R + 2 A^{\fraka b} L_{ab} 
+ C^{\fraka}{}_{\frakb \frakc} F^{\frakc}{}_{ab} A^{\frakb b}  
+ C^{\fraka}{}_{\frakb \frakc} f^\frakc(x) A^{\frakb}{}_{a}  
- C^{\fraka}{}_{\frakb \frakc} A^{\frakb b} \nabla_{b}A^{\frakc}{}_{a} + \nabla_{a}f^{\fraka}(x).
\label{CWEPotentialYM}
\end{equation}

\medskip
\noindent \textbf{Equation for the auxiliary field.} A suitable integrability
condition for the field $F^\fraka{}_{abc}$ can be obtained from its definition.
Using this and equation \eqref{UnphysicalYM3}, some manipulations yield
\[
3\nabla_{[d}F^\fraka{}_{|a|bc]} = -3 \Xi F^\fraka{}_{[b}{}^e d_{cd]ae} + 6 F^\fraka{}_{[bc}L_{d]a}
+ 6 g_{a[b}F^\fraka{}_c{}^e L_{d]e} - 3 C^\fraka{}_{\frakb\frakc}F^\frakc{}_{a[bc}A^\frakb{}_{d]} 
- 3 C^\fraka{}_{\frakb\frakc}\nabla_a A^\frakb{}_{[b} F^\frakc{}_{cd]}.
\]
Proceeding as in the case of the wave equation for $F^\fraka{}_{abc}$, as well
as using the Jacobi identity and definitions
\eqref{DefinitionZQYM1}-\eqref{DefinitionDerFaradayZQYM}, a lengthy calculation
results in
\begin{align}
\square F^\fraka{}_{abc} = & \ \tfrac{1}{2} F^{\mathfrak{a}}{}_{abc} R + 4 F^{\mathfrak{a} d}{}_{bc} L_{ad} + 2
F^{\mathfrak{b} d}{}_{bc} F^{\mathfrak{c}}{}_{ad} C^{\mathfrak{a}}{}_{\mathfrak{b}
\mathfrak{c}} -  F^{\mathfrak{c}}{}_{abc} f^{\mathfrak{b}}(x)
C^{\mathfrak{a}}{}_{\mathfrak{b} \mathfrak{c}}
 - F^{\mathfrak{d}}{}_{abc} A^{\mathfrak{b} d} A^{\mathfrak{c}}{}_{d} C^{\mathfrak{a}}{}_{\mathfrak{b} \frake}
C^{\frake}{}_{\mathfrak{c} \mathfrak{d}} \nonumber \\
& \ + \tfrac{1}{3} F^{\mathfrak{a}}{}_{bc}
\nabla_{a}R - 2 A^{\mathfrak{b} d} C^{\mathfrak{a}}{}_{\mathfrak{b} \mathfrak{c}}
\nabla_{d}F^{\mathfrak{c}}{}_{abc}
 - F^{\mathfrak{a}}{}_{d}{}^{e} d_{aebc}
\nabla^{d}\Xi -  F^{\mathfrak{a}}{}_{a}{}^{e} d_{bcde} \nabla^{d}\Xi + 2
\Xi F^{\mathfrak{a} de} \nabla_{e}d_{adbc} \nonumber \\
& \ - 4 \Xi F^{\mathfrak{a}d}{}_{[b}{}^{e}d_{|ad|c]e} 
- 2 \Xi F^{\mathfrak{a}}{}_{a}{}^{de}d_{[b|d|c]e}
- 4 F^{\mathfrak{a} d}{}_{a[b}L_{c]d} + 4 \Xi F^{\mathfrak{a}}{}_{[b}{}^{d}T_{c]da}
 + 4 \Xi F^{\mathfrak{a}}{}_{[b}{}^{d}T_{|ad|c]} \nonumber \\
& \ - \tfrac{1}{3}F^{\mathfrak{a}}{}_{a[b}\nabla_{c]}R 
+ 4 F^{\mathfrak{b}}{}_{a[b}{}^{d}F^{\mathfrak{c}}{}_{c]d}C^{\mathfrak{a}}{}_{\mathfrak{b}
\mathfrak{c}} - 4 F^{\mathfrak{a} d}{}_{[b}{}^{e}L_{|de}g_{a|c]}
 - 2 \Xi F^{\mathfrak{a} de}T_{[b|de}g_{a|c]} \nonumber \\
& \ - 4 F^{\mathfrak{a}}{}_{[b}{}^{d}d_{|ad|c]}{}^{e}\nabla_{e}\Xi - 2
F^{\mathfrak{a}}{}_{[b}{}^{d}d_{|a|}{}^{e}{}_{c]d}\nabla_{e}\Xi
 + 2 F^{\mathfrak{a} de}d_{ad[b|e|}\nabla_{c]}\Xi -  \tfrac{1}{3}
F^{\mathfrak{a}}{}_{[b}{}^{d}g_{|a|c]}\nabla_{d}R \nonumber \\
& \ - 2 F^{\mathfrak{a}de}g_{a[b}\nabla_{|d|}L_{c]e}.
\label{CWEDerFaradayYM}
\end{align}

In a similar manner to the two previous matter models, when equations
\eqref{CWEFaradayYM}, \eqref{CWEPotentialYM} and \eqref{CWEDerFaradayYM} are
coupled to the system of wave equations for the conformal fields, the $\square$
operator is to be replaced by its counterpart $\blacksquare$.

\subsubsection{Subsidiary equations}

The next step in the analysis of the Yang-Mills field is the introduction of
the corresponding subsidiary quantities and the consequent construction of
suitable geometric wave equations for them. For this purpose define the
following set of zero-quantities:
\begin{subequations}
\begin{eqnarray}
&& M^\fraka{}_a \equiv \nabla^b F^\fraka{}_{ba} + C^\fraka{}_{\frakb\frakc}
A^{\frakb b} F^\frakc{}_{ba} , \label{DefinitionZQYM1}\\
&& M^\fraka{}_{ab} \equiv \nabla_a A^\fraka{}_b - \nabla_b A^\fraka{}_a +
C^\fraka{}_{\frakb\frakc} A^\frakb{}_a A^\frakc{}_b -F^\fraka{}_{ab}, \label{DefinitionZQYMPotential}\\
&& M^\fraka{}_{abc} \equiv \nabla_{[a} F{}^\fraka{}_{bc]} + C^\fraka{}_{\frakb\frakc}
A^\frakb{}_{[a} F^\frakc{}_{bc]}, \label{DefinitionZQYM2} \\
&& Q^\fraka{}_{abc} \equiv F^\fraka{}_{abc} - \nabla_a F^\fraka{}_{bc}
-C^\fraka{}_{\frakb\frakc}A^\frakb{}_a F^\frakc{}_{bc}.
\label{DefinitionDerFaradayZQYM}
\end{eqnarray}
\end{subequations}
Notice that, unlike the Maxwell field analysis, an additional field
$M^\fraka{}_{ab}$ must be considered due to the introduction of the gauge
potential $A^\fraka{}_a$. Combining \eqref{DefinitionZQYM2} and
\eqref{DefinitionDerFaradayZQYM} an auxiliary relation is directly
obtained, namely
\begin{equation}
3 M^\fraka{}_{abc} + 3Q^\fraka{}_{[abc]} - 3F^\fraka{}_{[abc]} = 0.
\label{AdemIdentityYM}
\end{equation}
From these definitions, it follows that $M^\fraka{}_{abc}$ and $M^\fraka{}_{ab}$ possess the symmetries
\begin{equation}
M^\fraka{}_{abc} = M^\fraka{}_{a[bc]} = M^\fraka{}_{[ab]c} = M^\fraka{}_{[abc]}, 
\quad M^\fraka{}_{ab} = -M^\fraka{}_{ba}.
\end{equation}
Furthermore, direct calculations show that the Yang-Mills
zero-quantities satisfy the relations 
\begin{subequations}
\begin{eqnarray}
&& \nabla_a M^{\fraka a} = -C^\fraka{}_{\frakb\frakc} A^{\frakb a} M^\frakc{}_a
+\tfrac12 C^\fraka{}_{\frakb\frakc} F^{\frakb ab} M^\frakc{}_{ab}, 
\label{DivergenceMaYM}\\
&& \nabla^b M^\fraka{}_{ab} = M^\fraka{}_a, \label{DivergenceMabYM} \\
&& \nabla_a M^\fraka{}_{bc}{}^a = - \tfrac23\nabla_{[b}M^\fraka{}_{c]}
- \tfrac23 C^\fraka{}_{\frakb\frakc}A^\frakb{}_{[b}M^\frakc{}_{c]}
- C^\fraka{}_{\frakb\frakc}A^{\frakb a}M^\frakc{}_{abc}
- \tfrac23 C^\fraka{}_{\frakb\frakc}A^{\frakb a}Q^\frakc{}_{abc}
\nonumber \\
&& \hspace{1.9cm} - \tfrac23 C^\fraka{}_{\frakb\frakc} F^\frakb{}_{[b}{}^a M^\frakc{}_{c]a}.
\label{DivergenceMabcYM}
\end{eqnarray}
\end{subequations}

\medskip
\noindent
\textbf{Equation for $M^\fraka{}_{ab}$.} Consider the expression
$3\nabla_{[c}M^\fraka{}_{ab]}$.  Commuting covariant derivatives, substituting
expressions \eqref{DefinitionZQYM2}, \eqref{DefinitionDerFaradayZQYM} and
exploiting the Jacobi identity for the structure constants, the integrability
condition is obtained:
\begin{equation}
3\nabla_{[c}M^\fraka{}_{ab]} = - M^\fraka{}_{abc} - 3 C^\fraka{}_{\frakb\frakc}A^\frakb{}_{[a} M^\frakc{}_{bc]}.
\label{ICPotentialZQYM}
\end{equation}
Applying $\nabla^c$ to the last equation, a short calculation using equations
\eqref{DivergenceMaYM} and \eqref{DivergenceMabcYM} yields
\begin{align}
\square M^\fraka{}_{ab} = & \ 3 C^\fraka{}_{\frakb\frakc} A^\frakb M^\frakc{}_{abc} +
2 C^\fraka{}_{\frakb\frakc} A^{\frakb c} Q^\frakc{}_{cab} +\tfrac13 R M^\fraka{}_{ab}
-2 d_{acbd}M^{\fraka cd} - C^\fraka{}_{\frakb\frakc}f^\frakb M^\frakc{}_{ab} \nonumber \\
& - 2 C^\fraka{}_{\frakb\frakc} F_{[a}{}^c M^\frakc{}_{b]c}  + 2\nabla_c M^\fraka{}_{ab}{}^c
-C^\fraka{}_{\frakb\frakc} A^{\frakb c}\nabla_c M^\frakc{}_{ab}
-2 C^\fraka{}_{\frakb\frakc} M^\frakc{}_{c[a}\nabla^c A^\frakb{}_{b]}.
\label{WEZQYMPotential}
\end{align}

\medskip
\noindent
\textbf{Equation for $M^\fraka{}_a$.} Equation \eqref{DivergenceMabcYM}
constitutes an integrability condition for the field $M^\fraka{}_a$.  A
suitable wave equation can be obtained by first applying $\nabla^c$, commuting
derivatives and observing that $\nabla_c\nabla_a M^\fraka{}_b{}^{ac} =
\nabla_{[c}\nabla_{a]} M^\fraka{}_b{}^{ac}$.  Then, using definitions
\eqref{DefinitionZQYM1}-\eqref{DefinitionDerFaradayZQYM} along with
\eqref{DivergenceMaYM}, \eqref{DivergenceMabYM}, \eqref{ICPotentialZQYM}, the
Jacobi identity, and an appropriate substitution of \eqref{AdemIdentityYM}, a
long but straightforward computation results in:
\begin{align}
\square M^\fraka{}_b = & \ 2 L_{ba} M^{\mathfrak{a} a} + \tfrac{1}{6} R
M^{\mathfrak{a}}{}_{b} + 2 F^{\mathfrak{c}}{}_{ba} C^{\mathfrak{a}}{}_{\mathfrak{b}
\mathfrak{c}} M^{\mathfrak{b} a} - f^{\mathfrak{b}}(x)
C^{\mathfrak{a}}{}_{\mathfrak{b} \mathfrak{c}} M^{\mathfrak{c}}{}_{b}
 - A^{\mathfrak{b} a} A^{\mathfrak{c}}{}_{a} C^{\mathfrak{a}}{}_{\mathfrak{b} \frake}
C^{\frake}{}_{\mathfrak{c} \mathfrak{d}} M^{\mathfrak{d}}{}_{b} \nonumber \\
& -  \tfrac{3}{2} F^{\mathfrak{b} ac}
C^{\mathfrak{a}}{}_{\mathfrak{b} \mathfrak{c}} M^{\mathfrak{c}}{}_{bac} + 3
A^{\mathfrak{b} a} A^{\mathfrak{c} c} C^{\mathfrak{a}}{}_{\mathfrak{b} \mathfrak{d}}
C^{\mathfrak{d}}{}_{\mathfrak{c} \frake} M^{\frake}{}_{bac}
 + 2A^{\mathfrak{b} a} A^{\mathfrak{c} c} C^{\mathfrak{a}}{}_{\mathfrak{b} \mathfrak{d}}
C^{\mathfrak{d}}{}_{\mathfrak{c} \frake} Q^{\frake}{}_{cba} \nonumber \\ 
& - \tfrac{3}{2}
C^{\mathfrak{a}}{}_{\mathfrak{b} \mathfrak{c}} M^{\mathfrak{c}}{}_{bac} M^{\mathfrak{b}
ac} + 2 F^{\mathfrak{b} a}{}_{b}{}^{c} C^{\mathfrak{a}}{}_{\mathfrak{b} \mathfrak{c}}
M^{\mathfrak{c}}{}_{ac} - 2 C^{\mathfrak{a}}{}_{\mathfrak{b} \mathfrak{c}}
Q^{\mathfrak{b} a}{}_{b}{}^{c} M^{\mathfrak{c}}{}_{ac} \nonumber \\
& + F^{\mathfrak{c}}{}_{b}{}^{c} A^{\mathfrak{b} a} C^{\mathfrak{a}}{}_{\mathfrak{c}
\mathfrak{d}} C^{\mathfrak{d}}{}_{\mathfrak{b} \frake} M^{\frake}{}_{ac} - 2
A^{\mathfrak{b} a} C^{\mathfrak{a}}{}_{\mathfrak{b} \mathfrak{c}}
\nabla_{a}M^{\mathfrak{c}}{}_{b} + 2 A^{\mathfrak{b} a} C^{\mathfrak{a}}{}_{\mathfrak{b} 
\mathfrak{c}} \nabla_{c}Q^{\mathfrak{c}}{}_{ab}{}^{c} \nonumber \\
& - 3 C^{\mathfrak{a}}{}_{\mathfrak{b} \mathfrak{c}} M^{\mathfrak{c}}{}_{bac}
\nabla^{c}A^{\mathfrak{b} a} + 2 C^{\mathfrak{a}}{}_{\mathfrak{b} \mathfrak{c}}
Q^{\mathfrak{c}}{}_{abc} \nabla^{c}A^{\mathfrak{b} a}.
\label{WEZQYM1}
\end{align}

\medskip
\noindent
\textbf{Equation for $M^\fraka{}_{abc}$.} In a similar fashion to the
approach adopted for the electromagnetic
zero-quantity $M_{abc}$, and in order to simplify the
calculations, we introduce the Hodge dual of $M^\fraka{}_{abc}$:
\begin{equation}
M^{*\fraka}{}_a \equiv C^\fraka{}_{\frakb\frakc} F^{*\frakc}{}_{ba}A^{\frakb b} +
\nabla^b F^{*\fraka}{}_{ba} = \tfrac12 \epsilon_a{}^{bcd}M^\fraka{}_{bcd}.
\end{equation}
Here, the second equality has been obtained with help of
\eqref{FaradayYMDual} and \eqref{DefinitionZQYM2}. With this expression
we compute the divergence of $M^{*\fraka}{}_a$. Making use of
\eqref{DefinitionZQYMPotential} and the Jacobi identity, a calculation yields
\begin{eqnarray}
&& \nabla_a M^{*\fraka a} = -C^\fraka{}_{\frakb\frake}C^\frake{}_{\frakc\frakd}F^{*\frakd}{}_{ab}
A^{\frakb a}A^{\frakc b} - C^\fraka{}_{\frakb\frake}A^{\frakb a}M^{*\frakc}{}_a
+C^\fraka{}_{\frakb\frake}F^{*\frakc}{}_{ab}\nabla^b A^{\frakb a} \nonumber \\
&& \hspace{1.4cm} = - C^\fraka{}_{\frakb\frake}A^{\frakb a}M^{*\frakc}{}_a 
- \tfrac14 C^\fraka{}_{\frakb\frake} \epsilon_{ab}{}^{cd}F^{\frakb ab}M^\frakc{}_{cd}. \nonumber
\end{eqnarray}
In terms of non-dual objects this takes the form of an integrability condition:
\begin{eqnarray}
&& \epsilon^{abcd} \nabla_d M^{\fraka}{}_{abc} = 
C^{\fraka}{}_{\frakb \frakc}\epsilon^{abcd} A^{\frakb}{}_a M^{\frakc}{}_{bcd} 
+ \tfrac{1}{2} C^{\fraka}{}_{\frakb \frakc}\epsilon^{abcd} F^{\frakb}{}_{ab} M^{\frakc}{}_{cd} \nonumber \\
&& \hspace{-0.55cm} \iff 4\nabla_{[a}M^\fraka{}_{bcd]} = 4 C^{\fraka}{}_{\frakb \frakc} A_{[a}M^\fraka{}_{bcd]}
+ 2C^{\fraka}{}_{\frakb \frakc} F^{\frakb}{}_{[ab}M^{\frakc}{}_{cd]}.
\end{eqnarray}

Then, a suitable wave equation can be obtained applying $\nabla^d$ and
commuting derivatives. After a long calculation in which definitions
\eqref{DefinitionZQYM1}-\eqref{DefinitionDerFaradayZQYM}, equations
\eqref{AdemIdentityYM}-\eqref{ICPotentialZQYM} and the Jacobi identity
are employed, one finds that
\begin{align}
\square M^{\fraka}{}_{abc} = & \ \tfrac{1}{2} R M^{\fraka}{}_{abc} -
A^{\frakb d} A^{\frakc}{}_{d} C^{\fraka}{}_{\frakb \frakd}
C^{\frakd}{}_{\frakc \frake} M^{\frake}{}_{abc} - C^{\fraka}{}_{\frakb \frakc} 
f^\frakb(x) M^{\frakc}{}_{abc} - 2 A^{\frakb d} C^{\fraka}{}_{\frakb \frakc} \nabla_{d}M^{\frakc}{}_{abc} \nonumber \\
& - 6 \Xi d_{[a}{}^{d}{}_{b}{}^{e}M^{\fraka}{}_{c]de} - 6 L_{[a}{}^{d}M^{\fraka}{}_{bc]d}
 + 2 F^{\frakb d}{}_{[ab}C^{\fraka}{}_{|\frakb \frakc |}M^{\frakc}{}_{c]d} - 6
F^{\frakb}{}_{[a}{}^{d}C^{\fraka}{}_{|\frakb \frakc
|}M^{\frakc}{}_{bc]d} \nonumber \\
&  - 2 A^{\frakb d}C^{\fraka}{}_{\frakb \frakc}\nabla_{[a}Q^{\frakc}{}_{|d|bc]} 
 + 2 C^{\fraka}{}_{\frakb \frakc}Q^{\frakb d}{}_{[ab}\nabla_{c]}A^{\frakc}{}_{d} 
- 2 C^{\fraka}{}_{\frakb \frakc}Q^{\frakb d}{}_{[ab}M^{\frakc}{}_{c]d} \nonumber \\
& + F^{\frakb}{}_{[ab}A^{\frakc d}C^{\fraka}{}_{|\frakb\frakd}C^{\frakd}{}_{\frakc \frake |}M^{\frake}{}_{c]d}
 - 2 A^{\frakb}{}_{[a}A^{\frakc d}C^{\fraka}{}_{|\frakb \frakd}C^{\frakd}{}_{\frakc \frake}Q^{\frake}{}_{d|bc]}.
\label{WEZQYM2}
\end{align}

\medskip
\noindent
\textbf{Equation for $Q^\fraka{}_{abc}$.} Similar to the case for the Maxwell
field, a wave equation for $Q^\fraka{}_{abc}$ can be obtained by directly applying the
$\square$ operator to its definition. Since the identity used in the
deduction of equation \eqref{WEZQMaxwellAuxiliary} has the same form for the
Yang-Mills strength field, an analogous procedure can be followed.
A long computation gives:
\begin{align}
\square Q^\fraka{}_{abc} = & \ 6 L_{a}{}^{d} M^{\fraka}{}_{bcd} +
\tfrac{1}{2} R Q^{\fraka}{}_{abc} + 4 L_{a}{}^{d} Q^{\fraka}{}_{dbc}
- f^{\frakb}(x) C^{\fraka}{}_{\frakb \frakc} Q^{\frakc}{}_{abc}  - 2 F^{\frakb}{}_{a}{}^{d}
C^{\fraka}{}_{\frakb \frakc} Q^{\frakc}{}_{dbc} \nonumber \\ 
& - A^{\frakb d} A^{\frakc}{}_{d} C^{\fraka}{}_{\frakb \frakd}
C^{\frakd}{}_{\frakc \frake} Q^{\frake}{}_{abc} + 2
A^{\frakb}{}_{a} A^{\frakc d} C^{\fraka}{}_{\frakb \frakd}
C^{\frakd}{}_{\frakc \frake} Q^{\frake}{}_{dbc} + F^{\frakc}{}_{bc} A^{\frakb d} C^{\fraka}{}_{\frakc \frakd}
C^{\frakd}{}_{\frakb \frake} M^{\frake}{}_{ad} \nonumber \\
& - 2 F^{\frakc}{}_{bc} A^{\frakb d} C^{\fraka}{}_{\frakb \frakd}
C^{\frakd}{}_{\frakc \frake} M^{\frake}{}_{ad} + 2
C^{\fraka}{}_{\frakb \frakc} Q^{\frakc}{}_{dbc}
\nabla_{a}A^{\frakb d} + 4 A^{\frakb d} C^{\fraka}{}_{\frakb
\frakc} \nabla_{[a}Q^{\frakc}{}_{d]bc} \nonumber \\
& + 2 C^{\fraka}{}_{\frakb \frakc} M^{\frakb}{}_{a}{}^{d} \nabla_{d}F^{\frakc}{}_{bc} 
+ 4 \Xi F^{\fraka}{}_{[b}{}^{d}\Lambda_{c]ad}
 + 4 \Xi F^{\fraka}{}_{[b}{}^{d}\Lambda_{|a|c]d} - 2 \Xi
d_{[b}{}^{d}{}_{c]}{}^{e}Q^{\fraka}{}_{ade} \nonumber \\
& + 4 \Xi d_{a}{}^{d}{}_{[b}{}^{e}Q^{\fraka}{}_{|d|c]e} + 4 L_{a[b}M^{\fraka}{}_{c]}
 + 4 L_{[b}{}^{d}Q^{\fraka}{}_{|da|c]} + \Xi F^{\fraka de}\Lambda_{[b|de}g_{a|c]} \nonumber \\
& + 4 F^{\frakb}{}_{[b}{}^{d}C^{\fraka}{}_{|\frakb
\frakc}Q^{\frakc}{}_{a|c]d} + 4 L^{de}g_{a[b}Q^{\fraka}{}_{|d|c]e}.
\label{WEZQYMQ}
\end{align}

\subsubsection{Propagation of the gauge}

In this subsection we show the consistency of the introduction of the gauge
source functions $f^\fraka(x)$ into the analysis of the propagation of the
constraints for the Yang-Mills potential. For this purpose we introduce the
zero-quantity $P^\fraka$ defined as:
\begin{equation}
P^\fraka \equiv \nabla^a A^\fraka{}_a - f^\fraka(x).
\end{equation}
The computation of a wave equation for this field is straightforward:
first, a short calculation employing equations \eqref{CWEPotentialYM},
\eqref{DefinitionZQYM1}, \eqref{DefinitionZQYMPotential} and
\eqref{DivergenceMabYM} gives
\[
\nabla_a P^\fraka = -A^\frakb{}_b C^\fraka{}_{\frakb\frakc}P^\frakc
-M^\fraka{}_b +\nabla_a M^\fraka{}_b{}^a .
\]
From here, application of a further covariant derivative results directly in
\begin{equation}
\square P^\fraka = -f^\frakb C^\fraka{}_{\frakb\frakc}P^\frakc
+A^{\frakb a}C^\fraka{}_{\frakb\frakc}M^\frakc{}_a
-\tfrac12 F^{\frakb ab}C^\fraka{}_{\frakb\frakc}M^\frakc{}_{ab}
- A^{\frakb b}C^\fraka{}_{\frakb\frakc}\nabla_b P^\frakc. \label{WEZQYMGauge}
\end{equation}

\begin{remark}
{\em Geometric wave equations \eqref{WEZQYMPotential}, \eqref{WEZQYM1},
\eqref{WEZQYM2}, \eqref{WEZQYMQ} and \eqref{WEZQYMGauge} are homogeneous in
$M^\fraka{}_a, \ M^\fraka{}_{ab}, \ M^\fraka{}_{abc}, \ Q^\fraka{}_{abc}$,
$P^\fraka$, $\Lambda_{abc}$ and their first covariant derivatives.
Thus, if these fields vanish on an initial hypersurface
$\mathcal{S}_\star$, it can be guaranteed that there exists a unique solution
on a small enough slab of $\mathcal{S}_\star$ and this solution is the trivial
one.}
\end{remark}

\subsubsection{Summary}

The previous discussion about the Yang-Mills field coupled to the conformal
Einstein field equations leads to the following statement:

\begin{proposition}
The system of wave equations \eqref{ReducedWaveCFE1}-\eqref{ReducedWaveCFE5}
with energy-momentum tensor given by \eqref{EMTensorYM} coupled to wave
equations \eqref{CWEFaradayYM}, \eqref{CWEPotentialYM} and
\eqref{CWEDerFaradayYM} written in terms of the operator $\blacksquare$
is a proper quasilinear system of wave equations for the Einstein-Yang-Mills system.
\end{proposition}

\section{Applications}
\label{Section:Applications}

The purpose of this section is to provide a direct application of the
analysis of the evolution systems and subsidiary equations associated
to the conformal Einstein field equations coupled to tracefree
matter. Arguably, the simplest applications of our analysis to a
problem of global nature is that of the existence and stability of
de-Sitter like spacetimes. The original stability result of this type,
for vacuum perturbations, was carried in \cite{Fri86b}. For the sake
of conciseness of the presentation and given that the key technical
details have been discussed in the literature --- see
e.g. \cite{CFEBook}, Chapter 15 --- here we pursue a \emph{high-level}
presentation in the spirit of \cite{Fri15}.

\medskip
In order to present the result, it is recalled that one of the key features of
the conformal Einstein field equations is that they are regular up to the
conformal boundary. This property is also satisfied by the conformally invariant
scalar field equation, the Maxwell equations and the Yang-Mills equations.
Thus, they admit initial data prescribed on spacelike hypersurfaces describing
the conformal boundary of spacetime. In an analogous way to the Einstein field
equations, the metric conformal Einstein field equations admit a 3+1
decomposition with respect to a foliation of spacelike hypersurfaces. The
equations in this decomposition which are intrinsic to the spacelike
hypersurfaces are known as the \emph{conformal Einstein constraint equations}
--- see e.g. \cite{CFEBook}, Chapter 11. When evaluated on a spacelike
hypersurface representing the conformal boundary of a de Sitter-like spacetime,
these equations simplify considerably and a procedure to construct the
solutions to these equations is available --- see \cite{CFEBook}, Proposition
11.1 for the vacuum case; this result can be generalised to include tracefree
matter models. From the geometric side, the freely specifiable data in this
construction are given by the intrinsic metric of the conformal boundary and a
TT-tensor prescribing the electric part of the rescaled Weyl tensor. The
initial data obtained by this type of construction will be known as
\emph{asymptotic de Sitter-like initial data}. The component of the conformal
boundary where the asymptotic de Sitter-like data are prescribed can be either
the future or the past one. In the following, for convenience, we restrict the
discussion to the case of the past component of the conformal boundary. 

\medskip
 For asymptotic initial data sets of the type described in the
 previous paragraph one has the following result:

\begin{theorem}
\label{Theorem:Application}
Consider (past) asymptotic de-Sitter initial data for the Einstein field equations with
positive Cosmological constant coupled
to any of the following matter models:
\begin{enumerate}[(i)]
\item the conformally invariant scalar field,
\item the Maxwell field,
\item the Yang-Mills field.
\end{enumerate}
Then one has that: 
\begin{enumerate}[(a)]
\item  The initial data determine a unique, maximal, globally hyperbolic solution to the
Einstein field equations which admits a smooth de Sitter-like
conformal future extension.

\item The set of initial data sets leading to developments which admit
  smooth conformal extensions to both the future and past is an open
  set (in the appropriate Sobolev norm) of the set of asymptotic
  initial data.  

\end{enumerate} 
\end{theorem}

\begin{proof}
We only provide a sketch of the proof as the strategy is similar to
the one followed in the proof of the stability of the Milne spacetime
in \cite{GasVal15}. A version of the proof which uses first order
symmetric hyperbolic systems can be found in \cite{CFEBook}, Chapter 15. 

\medskip
The first main observation is that the conformal representation of the (vacuum)
de Sitter spacetime in terms of the embedding into the Einstein cylinder gives
rise to a solution to the conformal Einstein field equations. Coordinates
$(x)=(t,\underline{x})$ can be chosen so that the two components of the
conformal boundary are located at $t=\pm \tfrac{1}{2}\pi$. For this
representation the Ricci scalar takes the value $-6$ and the conformal factor
is given by $\mathring{\Xi}=\cos t$. In the following we denote by
$\mathring{\mathbf{u}}$ this solution to the conformal equations and by
$\mathring{\mathbf{u}}_\star$ its restriction to the hypersurface
$t=-\tfrac{1}{2}\pi$ which corresponds to the past conformal boundary
$\mathscr{I}^-$.  We will look for solutions to the conformal evolution
equations of the form $\mathbf{u}= \mathring{\mathbf{u}} + \breve{\mathbf{u}}$
with initial data given by $\mathbf{u}_\star= \mathring{\mathbf{u}}_\star +
\breve{\mathbf{u}}_\star$. The fields $ \breve{\mathbf{u}}$ and
$\breve{\mathbf{u}}_\star$ describe the (non-linear) perturbations. Substituting
this form of the solution into the evolution equations one obtains a system of
quasilinear equations for the components of $\breve{\mathbf{u}}$ which can be
schematically written as
\begin{equation}
\label{WaveEquationProof}
\big( \mathring{g}^{\mu\nu}(x)
+\breve{g}^{\mu\nu}(x;\breve{\mathbf{u}})\big) \partial_\mu\partial_\nu
\breve{\mathbf{u}} = \bmF(x;\breve{\mathbf{u}},\partial\breve{\mathbf{u}}).
\end{equation}
In the above expression $\mathring{g}^{\mu\nu}$ denote the
components of the contravariant metric on the Einstein cylinder. The
above equation is in the form for which the local existence and Cauchy
stability theory of quasilinear wave equations as given in, say,
\cite{HugKatMar77} applies. Initial data for the system
\eqref{WaveEquationProof} are of the form
$(\breve{\mathbf{u}}_\star,\partial_t\breve{\mathbf{u}}_\star)$. The
size of the initial data is encoded in the expression 
\[
\parallel (\breve{\mathbf{u}}_\star,\partial_t\breve{\mathbf{u}}_\star)\parallel_{\mathbb{S}^3,m}\equiv \parallel
\breve{\mathbf{u}}_\star\parallel_{\mathbb{S}^3,m}+ \parallel
\partial_t\breve{\mathbf{u}}_\star\parallel_{\mathbb{S}^3,m}
\]
 where $\parallel \phantom{}\parallel_{\mathbb{S}^3,m}$ denotes the standard
Sobolev norm of order $m\geq 4$ on a manifold which is topologically
$\mathbb{S}^3$. If the initial data
$(\breve{\mathbf{u}}_\star,\partial_t\breve{\mathbf{u}}_\star)$ are
sufficiently small then the contravariant metric on $\mathscr{I}^-$ given by
$\mathring{g}^{\mu\nu}(x_\star)
+\breve{g}^{\mu\nu}(x_\star;\breve{\mathbf{u}}_\star)$ is Lorentzian --- this
property is preserved in the evolution. Now, the background solution
$\mathring{\mathbf{u}}$ is well-defined and smooth on the whole of the Einstein
cylinder; in particular, up to $t=\pi$ for which one has that
$\mathring{\Xi}|_{t=\pi}=-1$. It follows from the Cauchy stability statements
in \cite{HugKatMar77} that if $\parallel
(\breve{\mathbf{u}}_\star,\partial_t\breve{\mathbf{u}}_\star)\parallel_{\mathbb{S}^3,m}$
is sufficiently small then the solution will exists up to $t=\pi$. By
restricting, if necessary, the size of the data one has that 
\[
\Xi|_{t=\pi}= -1 + \breve{\Xi}_{t=\pi} <0. 
\]
From the above observation it can be argued that the function
$\Xi=\mathring{\Xi}+\breve{\Xi}$ over the Einstein cylinder becomes zero on a
spacelike hypersurface which lies between the times $t=0$ and $t=\pi$. This
hypersurface corresponds to the future conformal boundary ($\mathscr{I}^+$)
arising from the data
$(\breve{\mathbf{u}}_\star,\partial_t\breve{\mathbf{u}}_\star)$ on
$\mathscr{I}^-$. 

\smallskip
Once the existence of a global solution to the evolution system has been
established, one makes use of the uniqueness of solutions to systems of
quasilinear wave equations to prove the propagation of the constraints. To this
end one observes that if the initial data satisfies the conformal constraints
at the past conformal boundary $\mathscr{I}^-$, then a calculation shows that
the  zero-quantities and their normal derivatives also vanish on
$\mathscr{I}^-$. As the subsidiary evolution system is homogeneous in the
zero-quantities, it follows that its unique solution must be the trivial (i.e.
vanishing) one. Thus, one has obtained a global solution to the conformal
Einstein field equations. From the general theory of the conformal Einstein
field equations --- see e.g. Proposition 8.1 in Chapter 8 of \cite{CFEBook} ---
this solution implies, in turn, a solution to the Einstein field equations with
positive Cosmological constant having de Sitter-like asymptotics.
\end{proof}

\begin{remark}
{\em The above theorem is a global stability result for the de Sitter spacetime
under perturbations involving a conformally invariant scalar field, a Maxwell
field or a Yang-Mills field as (trivially) the de Sitter spacetime can be
constructed from asymptotic initial data. Thus, for a suitably small
neighbourhood of asymptotic de Sitter data, all data in the neighbourhood give
rise to global solutions.}
\end{remark}

\begin{remark}
{\em The cases (ii) and (iii) --- the Maxwell and Yang-Mills fields, have been
studied using first order symmetric hyperbolic systems in \cite{Fri91}.
However, the case (i) --- the conformally invariant scalar field --- has, hitherto,
not been considered in the literature.}
\end{remark}

\begin{remark}
{\em The theory in \cite{HugKatMar77} is the analogue for systems wave
  equations of the theory for symmetric hyperbolic systems developed
  in \cite{Kat75}. A version of the key existence and Cauchy stability
result in \cite{HugKatMar77} given in the form used in Theorem
\ref{Theorem:Application} can be found in the Appendix of \cite{GasVal15}. }
\end{remark}

\section{Concluding remarks}
The global existence and stability result presented in Theorem
\ref{Theorem:Application} is the simplest application of the analysis of the
second order conformal evolution equations developed in this article. A further
application is to the construction of anti-de Sitter-like spacetimes with
tracefree matter models following the strategy implemented in \cite{CarVal18b}
--- this construction will be presented elsewhere \cite{CarVal19b}. The theory
developed in this article should also allow to obtain matter generalisation of
the existence results for characteristic initial value problems considered in
\cite{ChrPae13b}.

More crucially, the analysis in this article should also pave the road for
numerical simulations of spacetimes with tracefree matter in the conformal
setting. The use of the metric conformal Einstein equations in conjunction with
a coordinate gauge prescribed in terms of generalised wave condition provides a
formulation of the evolution equations for the conformal fields which can be
regarded as a (unphysical) \emph{reduced Einstein equation} with (unphysical)
matter described by the conformal factor, Friedrich scalar, Schouten tensor and
the rescaled Weyl tensor. Viewed in this way, one can readily adapt the
plethora of numerical know-how that has been developed in the numerical
simulations of the Einstein field equations. A further discussion of this idea
can be found in \cite{Fri03a}. 

\section*{Acknowledgements}
We thank Alfonso Garc\'{\i}a-Parrado for his help with certain aspects of the
computer algebra implementation of the calculations carried out in this
article. DAC thanks support granted by CONACyT (480147). We thank the anonymous
referee for comments and suggestions which have improved the presentation and
results of the article.

\end{document}